\documentclass[12pt,twoside]{article} 
\usepackage[T1]{fontenc} % font

\usepackage{amsthm, amsmath, amssymb, mathrsfs, upgreek, enumerate, mathtools, comment, natbib, graphicx, mathabx, booktabs, threeparttable, relsize, exscale, epstopdf, scrextend, multirow, xr-hyper, pdfpages, bm, caption, subcaption, tikz}
\usepackage[headheight=110pt,margin=1.25in]{geometry}
\usepackage[algoruled]{algorithm2e}
\usepackage[section]{placeins} % for \FloatBarrier command

\usepackage{fancyhdr, setspace} % headers and footers
\usepackage[pagebackref=false]{hyperref} % link colors
\hypersetup{
    colorlinks=true,
    citecolor=blue,
    filecolor=black,
    linkcolor=black,
    urlcolor=blue,
    bookmarksopen=true,
    pdfstartview=FitH
}

\newtheorem{theorem}{Theorem}
\newtheorem{assump}{Assumption}
\newtheorem{seclemma}{Lemma}[section]

\theoremstyle{definition}

\newtheorem{example}{Example}

\newtheorem{remark}{Remark}

\newtheorem{secdefinition}{Definition}[section]

\def\Snospace~{\S{}}
\makeatletter
\def\thm@space@setup{
  \thm@preskip=15pt \thm@postskip=15pt % controls spacing before and
}
\makeatother

\def\indep{\perp\!\!\!\perp}

\newcommand{\cov}{\text{Cov}}
\newcommand{\var}{\text{Var}}
\newcommand{\E}{{\bf E}}
\newcommand{\R}{\mathbb{R}}

\newcommand{\N}{\mathcal{N}}
\newcommand{\X}{\mathcal{X}}
\newcommand{\M}{\mathcal{M}}
\newcommand{\prob}{{\bf P}}
\newcommand{\plimarrow}{\stackrel{p}\longrightarrow}
\newcommand{\dlimarrow}{\stackrel{d}\longrightarrow}
\newcommand{\ind}{\bm{1}}
\providecommand{\abs}[1]{\lvert#1\rvert} 
\providecommand{\norm}[1]{\lVert#1\rVert}

\newcommand*{\medcap}{\mathbin{\scalebox{1.5}{\ensuremath{\cap}}}}

\let\emptyset\varnothing
\providecommand{\abs}[1]{\lvert#1\rvert} 
\providecommand{\norm}[1]{\lVert#1\rVert}

\renewcommand{\qed}{\hfill \mbox{\raggedright \rule{0.08in}{0.08in}}} % black QED box
\renewenvironment{proof}[1][\proofname]{{\noindent\sc#1. }}{\qed\vspace{15pt}} % "proof" small caps

\title{\bf\sc Designing Spatial Treatments} 

\author{Stefan Faridani\thanks{Department of Economics, Georgia Institute of Technology. E-mail: sfaridani6@gatech.edu.} \and Michael P.\ Leung\thanks{Department of Economics, University of California, Santa Cruz. E-mail: leungm@ucsc.edu.}}

\begin{document}
\maketitle
\onehalfspacing
 
\begin{abstract}

  {\sc Abstract.} Spatial treatments are interventions assigned to locations potentially distinct from those of the responding units. We study their optimal design under a general model in which a unit's response diminishes with distance to a treated site. Our estimand of interest is an ``uncontaminated'' effect equal to the average impact of a single intervention site over all hypothetical sites. We propose a novel design based on a Mat\'{e}rn point process which separates treatments by a distance of at least $r$. A larger choice of $r$ reduces bias by separating interventions but increases variance by reducing their numerosity. We choose $r$ to maximize the rate of convergence of a Horvitz-Thompson estimator and prove that this is minimax rate-optimal. We provide weak conditions under which the estimator is asymptotically normal and propose a variance estimator.

  \bigskip

  \noindent {\sc JEL Codes}: C14, C21, C93

  \noindent {\sc Keywords}: causal inference, experimental design, spatial treatments, interference
 
\end{abstract}

%----------------------------------------------------------------------
\section{Introduction}\label{sintro}
%----------------------------------------------------------------------

Place-based interventions have applications across a wide range of disciplines and policy settings. The construction of housing or urban amenities \citep{diamond2019wants}, hot-spot policing \citep{blattman2021place}, and the placement of micromobility vehicles such as e-bikes or scooters are all interventions that can affect units proximate to the site of treatment. Many ecological interventions are also fundamentally place-based, including species removal \citep{camargo2015experimental}, eco-engineering \citep{vozzo2024experiment}, and reserves such as marine protected areas \citep{ban2019well}. Following \cite{pollmann2023causal}, we refer to these interventions as {\em spatial treatments} to distinguish the setting from the typical potential outcomes model in which treatments are assigned to units rather than to sites potentially distinct from unit locations.

\cite{pollmann2023causal} and \cite{wang2023design} develop the first causal frameworks for evaluating the effects of spatial treatments. We build on their contributions by studying optimal design, but there are two key differences between our settings. First, their primary estimand of interest is a weighted average of the effects of multiple intervention sites with weights determined by the likelihoods that sites are treated. We refer to this as a ``contaminated'' effect because a unit may be responding to multiple treatments whose impacts cannot be disentangled. We instead consider an ``uncontaminated'' effect that compares the effect of a {\em single} intervention to the counterfactual of no intervention, averaged with respect to a chosen distribution $F$ over hypothetical intervention sites. The uncontaminated effect is design-independent, neither being a function of the distribution of treatments nor of a set of prespecified treatment locations, which enables the study of optimal design.

Second, both papers assume that the impact of a spatial treatment terminates after a fixed distance from the site. We consider a more general model of spillovers, allowing a unit's response to an intervention to gradually diminish as its distance to the site diverges. This induces a bias-variance trade-off in the placement of interventions. If interventions are well separated, this increases variance since fewer sites can be located in a given region. On the other hand, it reduces bias by better approximating an uncontaminated effect. The design problem is to locate interventions at the optimal level of separation to balance bias and variance, subject to locating more interventions in regions of higher probability under $F$.

We propose a novel design that generates intervention sites according to a Mat\'{e}rn type II point process. This entails drawing potential sites from $F$, dropping sites to ensure that the remainder are separated by a distance of at least $r$, and randomizing the remainder to treatment and control. The distance $r$ controls the bias-variance trade-off, and we choose it to maximize the rate of convergence of a Horvitz-Thompson estimator. We prove that, for any design and linear estimator, this rate is minimax-optimal over all potential outcomes. 

Under weak conditions, we establish that the Horvitz-Thompson estimator is asymptotically normal. Due to the Mat\'{e}rn-type design we employ, the data have a local dependence structure in that observations are independent if they are separated by a distance of $2r$ or greater. We may then apply a central limit theorem for local dependence, which is widely utilized in the interference literature \citep[e.g.][]{ogburn2024causal}. Unfortunately, the sufficient conditions are far too restrictive for our setting, delivering a suboptimal rate of convergence. We obtain optimal rates by using bounds on the Wasserstein distance tighter than the standard result. The bounds are more tedious to use compared to the off-the-shelf result, requiring computations involving fourth-order cross-moments.

\cite{faridani2023rate}, \cite{leung2022rate}, and \cite{leung2025cluster} study optimal designs for spatial interference, as opposed to spatial treatments. In their setting, interventions are assigned to units rather than to distinct spatial locations. Accordingly, they consider a different estimand, and the optimal design involves cluster-randomization rather than Mat\'{e}rn thinning. \cite{faridani2023rate} establish minimax rate-optimality. We employ similar techniques in the proof of our minimax result, but aspects of their argument require substantial modification. A final difference is the central limit theory, which is substantially different under spatial interference because the dependence structure is a hybrid of cluster- and near-epoch dependence, as opposed to local dependence in our setting. Thus, the normal approximation established in this paper relies on an entirely different set of techniques.

A related literature on bipartite interference also studies treatments assigned to intervention units distinct from response units \citep[e.g.][]{lu2025design,zigler2021bipartite,zigler2025bipartite}. Units are connected by a bipartite graph, and the edges determine which intervention units impact which response units. This may be viewed as a discrete analog of the continuous spatial treatments framework.

%----------------------------------------------------------------------
\section{Setup}\label{smodel}
%----------------------------------------------------------------------

Spatial treatments are randomly assigned to a finite number of sites, which are points in a compact set $\mathcal{R} \subseteq \R^d$. We denote the set of realized treatments by $\bm{S}$, which is observed by the analyst. Units are distributed in $\R^d$ according to an observed measure $\eta$. To represent unit responses to treatments, we adopt potential outcomes notation due to \cite{pollmann2023causal}. Let $\mathcal{P}$ be the set of all finite subsets of $\R^d$, and associate with each $x \in \R^d$ a (non-stochastic) potential outcome function $Y_x\colon \mathcal{P} \rightarrow \R$. For a unit positioned at $x \in \R^d$, $Y_x(S)$ denotes its potential outcome under the counterfactual that spatial treatments are given by $S \subseteq \mathcal{R}$ with only $Y_x(\bm{S})$ observed.

We suppose that a unit's response to a spatial treatment is diminishing in distance to the site. Let $B(x,r) = \{y \in \R^d\colon \norm{x-y} \leq r\}$, the ball of radius $r$ centered at $x$.

\begin{assump}[Spillovers]\label{aani}
  There exist $c,\gamma>0$ such that for all $s\in\mathbb{R}_+$,
  \begin{equation*}
    \sup_{x\in\R^d} \max\big\{\abs{Y_x(S) - Y_x(S')}\colon S,S' \subseteq \R^d \text{ finite}, S \cap B(x,s) = S' \cap B(x,s) \big\} \leq c\,s^{-\gamma}. 
  \end{equation*}
\end{assump}

\noindent That is, $Y_x(S)$ primarily depends on the subset of treatments relatively close to $x$. If $\gamma$ is small, then spillovers decay slowly with distance, and treatments can potentially influence more distant units. Assumption 3 of \cite{pollmann2023causal} and Assumption C3 of \cite{wang2023design} are stronger versions of \autoref{aani} that amount to replacing $c\,s^{-\gamma}$ with $c\, \ind\{s \leq d_0\}$. This means treatments do not impact units beyond a fixed radius $d_0$.

\autoref{aani} is an adaptation of the spatial interference model in Assumption 3 of \cite{leung2022rate} to the spatial treatments setup. As in the spatial interference literature, the coefficient $\gamma$ will play a central role in the optimal design because it provides a worst-case lower bound on the rate of decay. 

%--------------------------------------
\subsection{Causal Estimand}
%--------------------------------------

Let $\norm{\cdot}$ denote the Euclidean norm on $\R^d$, and let $\Delta \subseteq \R$ be an interval. For any $s \in \mathcal{R}$ and $S \subseteq \mathcal{R}$, define as in \cite{wang2023design}
\begin{equation*}
  \mu(s, S) = \frac{\int_x Y_x(S) \bm{1}\{\norm{x-s} \in \Delta\} \,\text{d}\eta}{\int_x \bm{1}\{\norm{x-s} \in \Delta\} \,\text{d}\eta}
\end{equation*}

\noindent with $\frac{0}{0} \equiv 0$. This is the average potential outcome of units within an annulus around $s$ defined by $\Delta$, under the counterfactual that the set of spatial treatments is $S$. It is possible to consider other functionals besides the average, but we focus our exposition on $\mu(s,S)$.

Define the site-specific causal effect
\begin{equation*}
  \tau(s) = \mu(s, \{s\}) - \mu(s, \emptyset) = \frac{\int_x \big( Y_x(\{s\}) - Y_x(\emptyset) \big) \bm{1}\{\norm{x-s} \in \Delta\} \,\text{d}\eta}{\int_x \bm{1}\{\norm{x-s} \in \Delta\} \,\text{d}\eta}.
\end{equation*}

\noindent This is the effect of a single treatment $s$ on units within the annulus around $s$. To aggregate up to an average effect over hypothetical locations $s$, the analyst specifies a continuous distribution $F$ supported on $\mathcal{R}$. The estimand of interest is
\begin{equation*}
  \theta_0 = \int_x \tau(x) \,\text{d}F,
\end{equation*}

\noindent which weights the site-specific effects $\{\tau(x)\}_{x\in\mathcal{R}}$ according to the scientific priorities of the analyst. For instance, if $\mathcal{R}$ contains urban and rural areas and only urban areas are of interest, then $F$ assigns zero probability mass to rural areas.

The primitive causal contrast used to construct $\theta_0$ is the ``uncontaminated'' effect $Y_x(\{s\}) - Y_x(\emptyset)$, so called because it is the causal effect of a single treatment. An analogous effect is defined in equation (1) of \cite{pollmann2023causal}, but this is not the focus of his paper. The primary estimands studied by \cite{pollmann2023causal} and \cite{wang2023design} instead aggregate comparisons of the form $\E[Y_x(\bm{S} \cup \{s\}) - Y_x(\bm{S} \backslash \{s\})]$. We call these ``contaminated'' in the sense that they involve multiple spatial treatments. Furthermore, the contaminated comparisons are design-dependent due to the expectation over $\bm{S}$, which precludes a study of optimal design. 

Both papers assume that $\bm{S}$ is supported on a fixed, prespecified set of sites which are the only locations that can be treated. To define our estimand, the analyst does not prespecify exact sites, which would correspond to $F$ with discrete support, but rather broader regions of interest according to $F$ with continuous support. The design challenge is to treat regions with higher weight under $F$ while also minimizing contamination due to sites placed in close proximity.

%--------------------------------------
\subsection{Design}
%--------------------------------------

In order to identify an uncontaminated effect, we seek to randomize treatment to locations that are sufficiently geographically separated. We employ the following Mat\'{e}rn type II ``thinning'' process. Let $\X_n = \{X_i\}_{i=1}^n \stackrel{iid}\sim F$ be a set of $n$ {\em proto-sites} and $\{U_i\}_{i=1}^n \stackrel{iid}\sim \mathcal{U}([0,1])$ be {\em marks} associated the proto-sites, generated independently of $\X_n$. Let
\begin{equation*}
  \M_r = \left\{ i \in \{1,\ldots,n\}\colon U_i \leq U_j \,\,\forall j \text{ s.t. } X_j \in B(X_i,r) \right\}.
\end{equation*}

\noindent That is, we obtain $\M_r$ from $\X_n$ by dropping sites with a higher mark than any other site in its $r$-ball. Then each pair of elements of $\M_r$ is separated by a distance of at least $r$.

Call site $i$ a {\em potential site} if $i \in \M_r$. We randomize potential sites into treatment with probability $p \in (0,1)$. Draw i.i.d.\ treatment assignments $\{D_i\}_{i=1}^n \stackrel{iid}\sim \text{Ber}(p)$ independently of $\{X_i\}_{i=1}^n$. The set of spatial treatments consists of potential sites that are randomized to treatment:
\begin{equation*}
  \bm{S} = \big\{X_i\colon i \in \M_r, D_i=1\big\}.
\end{equation*}

%--------------------------------------
\subsection{Estimators}
%--------------------------------------

Let $\rho_i = \prob(i \in \M_r \mid \X_n)$. By design, any proto-site in the $r$-ball of $i$ is equally likely to have the lowest mark, so $\rho_i = \abs{B(X_i,r) \cap \X_n}^{-1}$. We estimate $\theta_0$ using the Horvitz-Thompson estimator 
\begin{equation*}
  \hat\theta = \frac{1}{n} \sum_{i=1}^n \mu(X_i,\bm{S}) \frac{\bm{1}\{i \in \M_r\}}{\rho_i} \left( \frac{D_i}{p} - \frac{1-D_i}{1-p} \right),
\end{equation*}

\noindent which simply compares treated and untreated potential sites. Inverse probability weighting by $\rho_i$ adjusts for the possibility that certain proto-sites are more likely to be selected as potential sites due to a paucity of neighboring sites.

For a given number of proto-sites $n$, larger $r$ results in fewer potential sites, which results in a smaller effective sample size $\abs{\M_r}$ and hence larger variance. On the other hand, larger $r$ generates more separation between potential sites, which should aid in identifying the uncontaminated effect. Hence, the choice of $r$ determines a bias-variance trade-off.

The summands of $\hat\theta$ are not independent due to spillovers in $\mu(X_i,\bm{S})$ and correlation across the indicators $\bm{1}\{i \in \M_r\}$. Under the ``undersmoothed design'' discussed in \autoref{spi}, the correlation in the indicators is first-order. These have a local dependence structure in that $\bm{1}\{i \in \M_r\} \indep \bm{1}\{j \in \M_r\} \mid \X_n$ if $\norm{X_i-X_j} > 2r$. We therefore consider the natural variance estimator that accounts for covariances between sites up to distance $2r$:
\begin{multline*}
  \hat\sigma^2 = \frac{1}{n^2} \sum_{i=1}^n \sum_{j=1}^n \ind\{ \norm{X_i-X_j} \leq 2r \} (Z_i - \hat\theta) (Z_j - \hat\theta), \\ \text{where}\quad Z_i = \mu(X_i,\bm{S}) \frac{\bm{1}\{i \in \M_r\}}{\rho_i} \left( \frac{D_i}{p} - \frac{1-D_i}{1-p} \right).
\end{multline*}

%----------------------------------------------------------------------
\section{Asymptotic Theory}\label{sclt}
%----------------------------------------------------------------------

We study the asymptotic properties of the estimators by taking limits along a sequence of spatial regions such that $\text{Vol}(\mathcal{R})$ diverges with the number of proto-sites $n$.\footnote{For any $A \subseteq \R^d$, $\text{Vol}(A) = \int_{x \in A} \,\text{d}\mu$ where $\mu$ is the Lebesgue measure.} This is intended to approximate a relatively large spatial region, which in practice is necessary to construct sufficiently separated spatial treatments. If a region is small in the sense that its volume is fixed with respect to $n$, then every unit lies within a fixed distance of every treatment, so asymptotically negligible bias is only achieved in the case of no spillovers beyond a sufficiently short distance. 

Let $\mathcal{R}_0$ be a compact subset of $\R^d$ and $\{\lambda_n\}$ be any sequence of positive numbers such that
\begin{equation*}
  \lambda_n = O(n^{1/d}) \quad\text{and}\quad \lambda_n \rightarrow \infty.
\end{equation*}

\noindent We take 
\begin{equation*}
  \mathcal{R} = \big\{\lambda_n x\colon x \in \mathcal{R}_0\big\}.
\end{equation*}

\noindent Then the volume of $\mathcal{R}$ is order $\lambda_n^d$, and slow rates of divergence for $\lambda_n$ correspond to smaller spatial regions. Note that $\lambda_n$ may diverge at an arbitrarily slow rate.

Let $f$ be a density function and $\{\tilde X_i\}_{i=1}^n \stackrel{iid}\sim f$. For any $n\in\mathbb{N}$, we take proto-sites to be
\begin{equation*}
  X_i = \lambda_n \tilde X_i,
\end{equation*}

\noindent so that $F$ has density $\lambda_n^{-d}f(\lambda_n^{-1}x)$. In special case where $\lambda_n$ is exactly of order $n^{1/d}$, this corresponds to an increasing domain setup since the proportion of proto-sites lying within a fixed ball about $X_i$ is asymptotically constant. When $\lambda_n$ diverges more slowly, this corresponds to infill-increasing asymptotics where the number within a fixed ball diverges with $n$.

We set the radius of the Mat\'{e}rn thinning design to be
\begin{equation*}
  r \equiv r_n = \beta_n \lambda_n \quad\text{where}\quad \beta_n \rightarrow 0 \quad\text{and}\quad \lambda_n^{-1} = o(\beta_n). 
\end{equation*}

\noindent The rates imply $r_n \rightarrow \infty$, which is needed to ensure that the bias of $\hat\theta$ is asymptotically negligible. Since $\lambda_n = O(n^{1/d})$, they also imply $n\beta_n^d \rightarrow \infty$.

To summarize, we take limits along a sequence such that $r$, $\mathcal{R}$, and $F$ change with $n$ under the aforementioned rate restrictions. We also allow potential outcomes to change with $n$, subject to satisfying \autoref{aani} {\em with $c$ and $\gamma$ being universal constants that do not depend on $n$}. All other quantities are held fixed. We also maintain the following regularity conditions. 

\begin{assump}[Bounded Response]\label{aboundY}
  There exists $c_0>0$ such that $\int_x \bm{1}\{\norm{x-s} \in \Delta\} \,\text{d}\eta > c_0$ for all $s \in \R^d$. There exists $M < \infty$ such that $\int_{x \in B} \abs{Y_x(S)} \,\text{d}\eta < M$ for any bounded $B \subseteq \R^d$, finite $S \subseteq \mathcal{R}$, and $n\in\mathbb{N}$.
\end{assump}

\noindent This implies that $\mu(s,S)$ is bounded uniformly over the sequence.

\begin{assump}[Density]\label{adensity}
  On its support, $f$ is continuous, bounded, and bounded away from zero. Furthermore, its support is bounded with nontrivial volume and a Lipschitz boundary.
\end{assump}

\noindent A Lipschitz boundary (see \autoref{dlipbd}) is a mild regularity condition requiring that local segments of the boundary correspond to graphs of Lipschitz functions. This allows for a nonsmooth boundary.

%--------------------------------------
\subsection{Normal Approximation}
%--------------------------------------

Define the infeasible estimator
\begin{equation*}
  \hat\theta^* = \frac{1}{n} \sum_{i=1}^n \frac{\bm{1}\{i \in \M_{r_n}\}}{\rho_i} \left( \frac{\mu(X_i,\{X_i\})D_i}{p} - \frac{\mu(X_i,\{\emptyset\})(1-D_i)}{1-p} \right),
\end{equation*}

\noindent which is unbiased for $\theta_0$. Decompose
\begin{equation}
  \hat\theta - \theta_0 = \big(\hat\theta - \hat\theta^*\big) + \big(\hat\theta^* - \E[\hat\theta^* \mid \X_n]\big) + \big(\E[\hat\theta^* \mid \X_n] - \theta_0\big), \label{decomp}
\end{equation}

\noindent and note that $\E[\hat\theta^* \mid \X_n] = n^{-1} \sum_{i=1}^n \tau(X_i)$. The first term of \eqref{decomp} is the bias term, which will be shown to be $O_p(r_n^{-\gamma})$. The third term is mean zero with variance $n^{-1}\var(\tau(X_1))$, which is second-order relative to the other terms. It remains to derive the asymptotic distribution of the second term, which is an average of variables
\begin{multline}
  Z_i^* = \frac{\bm{1}\{i \in \M_{r_n}\}}{\rho_i} \left( \frac{\mu(X_i,\{X_i\})D_i}{p} - \frac{\mu(X_i,\{\emptyset\})(1-D_i)}{1-p} \right) \\ - \big( \mu(X_i,\{X_i\}) - \mu(X_i,\emptyset) \big). \label{Zistar}
\end{multline}

\noindent Let $\sigma_n^2 = \var(\hat\theta^* \mid \X_n)$, the conditional variance of the second term. We assume it is asymptotically nondegenerate in the following sense.

\begin{assump}[Non-degeneracy]\label{anondeg}
  $\liminf_{n\rightarrow\infty} \beta_n^{-d} \sigma_n^2 > 0$ a.s.
\end{assump}

\noindent The normalization by $\beta_n^{-d}$ is natural because $n^{-2} \sum_{i=1}^n \var(Z_i^* \mid \X_n) = O(\beta_n^d)$ by \autoref{lovar}. To violate \autoref{anondeg}, the sum of the covariance terms needs to be negative and large enough in magnitude to offset the sum of the variances. Later we will derive a closed-form expression for the limit variance under this scaling.

\begin{theorem}\label{tclt}
  Under Assumptions \ref{aani} and \ref{aboundY}, $\abs{\hat\theta - \hat\theta^*} = O_p(r_n^{-\gamma})$. Additionally, under Assumptions \ref{adensity} and \ref{anondeg}, $(\hat\theta^* - \theta_0) / \sigma_n \dlimarrow \N(0,1)$, so $\abs{\hat\theta^* - \theta_0} = O_p(\beta_n^{d/2})$.
\end{theorem}

\noindent The first claim bounds the asymptotic order of the bias term, showing that it decreases with the Mat\'{e}rn radius $r_n$. The second shows that the remainder is asymptotically normal, centered at $\theta_0$, and converges at rate $\beta_n^{d/2}$. Since $r_n = \beta_n \lambda_n$, $\beta_n$ determines a bias-variance trade-off. We discuss the rate-optimal choice in the next subsection.

To prove the second claim, we observe that the ``data'' $\{Z_i^*\}_{i=1}^n$ have a local dependence structure due to the indicators $\ind\{i \in \M_{r_n}\}$: conditional on $\X_n$, $Z_i^* \indep Z_j^*$ if $\norm{X_i-X_j} > 2r_n$. We may then apply the local dependence central limit theorem in \autoref{lstein}. This is not the usual result employed in the literature, as it uses tighter bounds on the Wasserstein distance. These bounds are obtained from an intermediate step in the proof of the usual result and are not typically used directly since they are inconvenient, and the looser bounds suffice for most applications. In our setting, however, the looser bounds result in suboptimal rates.

\begin{remark}
  Standard formulations of the local dependence central limit theorem, for example \cite{ross2011fundamentals} Theorem 3.6 and \cite{penrose2003} Theorem 2.4, require an asymptotic bound on $\Psi^2 (n\sigma_n)^{-3} \sum_{i=1} \E[\abs{Z_i^*}^3 \mid \X_n]$ where $\Psi = \max_i \sum_{j=1}^n \ind\{\norm{X_i-X_j} \leq 2r_n\}$ is the maximum degree of the dependency graph. This is significantly more conservative than the bound in \autoref{lstein}. To see this, note that $\Psi = O_p(n\beta_n^d)$ by \autoref{ldt}, $\liminf_n \beta_n^{d/2}\sigma_n > 0$ by \autoref{anondeg}, and $\E[\abs{Z_i^*}^3 \mid \X_n] = O_p( (n\beta_n^d)^2 )$ as shown prior to \eqref{Norder}. Therefore,
  \begin{equation*}
    \frac{\Psi^2}{(n\sigma_n)^3} \sum_{i=1} \E[\abs{Z_i^*}^3 \mid \X_n] = O_p\left( \frac{(n\beta_n^d)^2}{(n\beta_n^{d/2})^3} \right) \cdot n \cdot O_p(  (n\beta_n^d)^2 ) = O_p(n^2\beta_n^{5d/2}).
  \end{equation*}

  \noindent Whereas \autoref{tclt} only requires $\beta_n \rightarrow 0$, the above expression is only $o_p(1)$ if $\beta_n = o(n^{-4/(5d)})$. In order for the rate-optimal choice of $\beta_n$ (discussed in the next subsection) to satisfy this requirement in the increasing domain case of $\lambda_n = n^{1/d}$, we need $\gamma > 2d$, whereas \autoref{tclt} holds for any $\gamma > 0$.
  % Recall that $\beta_n = n^{-1/(d(1+d/\gamma/2)}$ is the optimal rate, so we need
  % 1/(d(1+d/\gamma/2) > 4/(5d) \Longleftrightarrow 5/(1+d/\gamma/2) > 4 \Longleftrightarrow 5/4 > 1+d/\gamma/2 \\
  % 0.25 > d/\gamma/2 \Longleftrightarrow 0.5 > d/\gamma \Longleftrightarrow \gamma > 2d.
\end{remark}

%--------------------------------------
\subsection{Minimax Rate}\label{sminmax}
%--------------------------------------

By \autoref{tclt}, the squared bias is order $r_n^{-2\gamma}$ while the variance is order $\beta_n^d$. Equating the two and solving for $\beta_n$, we obtain the optimal rate
\begin{equation}
  \beta_n = \lambda_n^{-\frac{1}{1+0.5d/\gamma}}. \label{optr}
\end{equation}

\noindent Under \eqref{optr}, the rate of convergence of $\hat\theta$ is $\lambda_n^{-\frac{d}{2 + d/\gamma}}$, which we will show is minimax. In the increasing domain case where $\lambda_n = n^{1/d}$, this rate equals $n^{-\frac{1}{2 + d/\gamma}}$, which matches the minimax rate in \cite{faridani2023rate}. In this sense, rate-optimal design of spatial treatments is as difficult a problem as rate-optimal design for spatial interference. In the increasing domain case, $\hat\theta$'s rate of convergence is always below the parametric rate but approaches $n^{-1/2}$ as $\gamma$ increases, meaning as spillovers become less persistent.

\begin{remark}\label{r.n}
  An interesting feature of the rates is that they depend directly on $\lambda_n^d$, the asymptotic order of $\text{Vol}(\mathcal{R})$, but not on $n$. Thus for a given spatial region, an unlimited budget for experimentation to increase the number of proto-sites does not translate to improved power. This is because the Mat\'{e}rn thinning design discards excess proto-sites until the remainder is well separated, so adding more sites by increasing $n$ eventually only results in them being discarded. Formally, choosing larger $n$ does not increase the effective sample size $\E[\abs{\M_r}] \leq n \E[\rho_i]$. By \autoref{ldt}, this is $O(\beta_n^{-d})$ and hence is determined by the Mat\'{e}rn radius and not $n$ (provided $n$ is large enough).
\end{remark}

We next show that $\hat\theta$ with the rate-optimal choice of $\beta_n$ for the Mat\'{e}rn thinning design achieves the minimax rate of convergence. Formally, for any design, distribution of units $\eta$, and linear estimator, there exist potential outcomes such that the estimator's rate of convergence is no faster than $\lambda_n^{-\frac{d}{2 + d/\gamma}}$. By ``linear estimators'' we mean those that can be written as
\begin{equation}\label{eq:linear_estimators}
  \hat{\theta}_\omega \equiv \int_x \omega_x(\bm{S}) Y_x(\bm{S})\text{d}\eta
\end{equation}

\noindent for weights $\omega_x(\bm{S})$ that may depend on $x$, $\bm{S}$, and $n$ but not on $Y_x(\cdot)$. 

\begin{theorem}\label{thm:minimax_rate}
  Let $\hat{\theta}_\omega$ be any estimator of the form \eqref{eq:linear_estimators}. Fix any sequence $b_n\rightarrow\infty$, sequence of designs (i.e.\ distributions of $\bm{S}$), and distribution of units $\eta$ satisfying \autoref{aboundY}. There exists a sequence of potential outcomes satisfying Assumptions \ref{aani} and \ref{aboundY} and $\delta>0$ such that:
  $${\limsup_{n\to\infty} \prob\big(b_n\lambda_n^{\frac{d}{2+d/\gamma}} \abs{\hat{\theta}_\omega - \theta_0} > \delta\big) >0.}$$
\end{theorem} 

%--------------------------------------
\subsection{Variance}
%--------------------------------------

The next result provides sufficient conditions under which the limit variance has a closed-form expression. 

\begin{assump}\label{astat}
  There exist continuous functions $\mu^1, \mu^0 : \mathrm{supp}(f) \to \R$ that do not depend on $n$ such that $\mu(X_1, \{X_1\}) = \mu^1(\tilde X_1)$ and $\mu(X_1,\emptyset) = \mu^0(\tilde X_1)$ a.s.\ for any $n$.
\end{assump}

\noindent We state the result before discussing the assumption. Let $\omega_d = \text{Vol}(B(0,1))$ and $v_d(\delta) \equiv \text{Vol}(B(0,1) \cap B(\delta e_1, 1))$ where $e_1 = (1,0,0,\ldots,0)$ is the first standard basis vector in $\R^d$. 

\begin{theorem}\label{tlimvar}
  Suppose $n\beta_n^d / \log(n) \rightarrow c$ for some sufficiently large $c \in (0,\infty]$. Under Assumptions \ref{aboundY}, \ref{adensity}, and \ref{astat},
  \begin{multline*}
    \beta_n^{-d} \sigma_n^2 \stackrel{a.s.}\longrightarrow \left( \frac{\omega_d}{p(1-p)}\int_x \big((1-p)\mu^1(x)+p\mu^0(x)\big)^2 f(x)^2 \,\text{d}x \right. \\ \left. + \int_{\norm{z}\in (1,2)} \frac{v_d(\norm{z})}{2\omega_d-v_d(\norm{z})} \,\text{d}z \int_x (\mu^1(x)-\mu^0(x))^2 f(x)^2 \,\text{d}x \right). 
  \end{multline*}
\end{theorem}

\noindent The limit is strictly positive whenever $(1-p)\mu^1(\tilde{X}_1)+p\mu^0(\tilde{X}_1)$ or $\mu^1(\tilde{X}_1)-\mu^0(\tilde{X}_1)$ are non-degenerate. The requirement $n\beta_n^d / \log(n) \rightarrow c$ is not restrictive because if $n\beta_n^d$ were of logarithmic order, then since $\lambda_n = O(n^{1/d})$, the bias would have to shrink at a poor logarithmic rate. The most important requirement of \autoref{astat} is $\mu(X_1, \{X_1\}) = \mu^1(\tilde X_1)$ and similarly for $\mu^0$. We next provide some examples.

\begin{example}
  The requirement holds if potential outcomes scale-normalize the inputs such that $Y_x(S) = g(x/\lambda_n, S/\lambda_n)$ for some $n$-independent function $g$. This says that responses vary at the scale of the study region $\lambda_n$. For instance, suppose $\mathcal{R} = \lambda_n \cdot [-1,1]^2$ and $Y_x(S) = \sum_{s\in S} w_{xs} \kappa_s$ where $w_{xs}$ is some continuous, decaying function of $\norm{x-s}$ (to satisfy \autoref{aani}). Then $Y_x(\{s\}) = w_{xs} \kappa_s$. The requirement holds if $\kappa_s$ is a continuous function of $s/\lambda_n$, for example $\kappa_s = \abs{s_1/\lambda_n}$ where $s_1$ is the first coordinate of $s$. This says that responses are increasing in the spatial treatment's proximity to the eastern or western boundary, where distance is measured as a proportion of the region's width. 
\end{example}

\begin{example}
  The requirement holds if potential outcomes are {\em translation-invariant} in that $Y_x(S) = Y_{x+z}(S+z)$ for any $z$. This means only relative distances matter in the model, not absolute locations, which is a strong spatial stationarity condition.
\end{example}

\begin{remark}
  Given \autoref{aboundY}, $\mu^1(s), \mu^0(s)$ are continuous if $\eta$ is absolutely continuous with respect to the Lebesgue measure, $s \mapsto Y_x(\{s\})$ is a continuous function, and $\sup_{s \in \text{supp}(f)} \abs{Y_x(\{s\})}$ is an $\eta$-integrable function of $x$.
\end{remark}

Let $\bar{\sigma}_n^2 = \var(\hat\theta^*)$, the unconditional variance. The last result characterizes the bias of the variance estimator and provides conditions under which $\hat\sigma^2$ is asymptotically conservative, both for the conditional and unconditional variance.

\begin{theorem}\label{tvar}
  Under Assumptions \ref{aani}--\ref{anondeg}, $(\hat\sigma^2-\mathcal{B}_n) / \sigma_n^2 \plimarrow 1$ for 
  \begin{equation*}
    \mathcal{B}_n = \frac{1}{n^2} \sum_{i=1}^n \sum_{j\neq i} \ind\{\norm{X_i-X_j} \leq 2r_n\} (\tau(X_i) - \bar\tau) (\tau(X_j) - \bar\tau), \quad \bar\tau = \frac{1}{n} \sum_{i=1}^n\tau(X_i).
  \end{equation*}

  \noindent If additionally $n\beta_n^d / \log(n) \rightarrow c$ for some sufficiently large $c \in (0,\infty]$, then $(\hat\sigma^2-\mathcal{B}_n) / \bar\sigma_n^2 \plimarrow 1$. Under Assumptions \ref{aboundY} and \ref{adensity}, if there exists $\tau_\infty(\cdot)$ such that $\tau(X_1) \stackrel{a.s.}\longrightarrow \tau_\infty(\tilde{X}_1)$, then for $\bar{\tau}_\infty = \int_x \tau_\infty(x) f(x) \,\text{d}x$,
  \begin{equation}
    \mathcal{B}_n \plimarrow \mathrm{Vol}(B(0,2)) \int_x (\tau_\infty(x) - \bar{\tau}_\infty)^2 f(x)^2 \,\text{d}x. \label{Bconv}
  \end{equation}
\end{theorem}

\noindent A sufficient condition for $\tau(X_1) \stackrel{a.s.}\longrightarrow \tau_\infty(\tilde{X}_1)$ is \autoref{astat}, which implies $\tau(X_1) = \mu^1(\tilde{X}_1)-\mu^0(\tilde{X}_1)$, so \eqref{Bconv} holds for $\tau_\infty(x) = \mu^1(x)-\mu^0(x)$.

At first glance, the integral \eqref{Bconv} has the form of a treatment effect heterogeneity term $\var(\tau_\infty(\tilde{X}_1))$, which is the usual reason for conservative variance estimation in design-based settings. In fact, the integral does not generally equal $\var(\tau_\infty(\tilde{X}_1))$ because the density $f$ is squared in the expression. Indeed, the contribution of $\var(\tau_\infty(\tilde{X}_1))$ to $\hat\sigma^2$ ends up being second-order, and \eqref{Bconv} is the {\em covariance} contribution from pairs of spatial treatments that are within distance $2r_n$. This can be seen from the inner sum in the expression of $\mathcal{B}_n$, which includes only covariance terms. The reason for the form of the limit is that $\norm{X_i-X_j} \leq 2r_n$ if and only if $\norm{\tilde{X}_i-\tilde{X}_j} \leq 2\beta_n$, which tends to zero, so \eqref{Bconv} is the total covariance between site pairs whose scaled locations are similar.

%----------------------------------------------------------------------
\section{Practical Implementation}\label{spi}
%----------------------------------------------------------------------

To implement our design and estimator, there are several parameters to specify. The first is the worst-case rate of decay $\gamma$ for spillovers, which we propose to calibrate as follows. Suppose spillovers decay exactly like $s^{-\gamma}$ for some unknown $\gamma$, and suppose the interval $\Delta$ in $\theta_0$ is given by $\Delta = [L, U]$. Then $100(U/L)^{-\gamma} \leq P$ where $P \in [0,100]$ is an upper bound on how the percentage of the spillover at distance $L$ that remains at distance $U$. Solving for $\gamma$ results in 
\begin{equation*}
  \gamma \geq -\frac{\log (P/100)}{\log (U/L)},
\end{equation*}

\noindent so the right-hand side can be taken as the worst-case rate of decay. The idea is that it is easier in practice to form a worst-case prior over $P$ rather than $\gamma$ directly.

Second is the radius of the Mat\'{e}rn thinning design $r = \beta_n \lambda_n$. Since $\text{Vol}(\mathcal{R})$ is of order $\lambda_n^d$, one can take 
\begin{equation*}
  \lambda_n = \text{Vol}(\mathcal{R})^{1/d}.
\end{equation*}

\noindent Choosing $\beta_n$ according to \eqref{optr} is rate-optimal but results in an asymptotic bias. As in \cite{leung2022rate}, we utilize an ``undersmoothed design'' to justify the standard confidence interval $\hat\theta \pm 1.96 \hat\sigma$. This corresponds to choosing 
\begin{equation}
  \beta_n = \lambda_n^{-\frac{1}{1+0.5d/\tilde\gamma}} \label{betaunder}
\end{equation}

\noindent for some $\tilde\gamma < \gamma$. In other words, we specify a more conservative choice for the spillover rate of decay when constructing $r$ to remove the asymptotic bias. Then given $r$, the analyst may consider estimands $\theta_0$ such that the endpoints of the annulus interval $\Delta$ are small relative to $r$. 

Third is the number of proto-sites $n$. As discussed in \autoref{r.n}, the number of potential sites $\sum_{i=1}^n \ind\{i \in \M_r\}$ does not increase with $n$ once it is sufficiently large. We therefore suggest simulating $\sum_{i=1}^n \ind\{i \in \M_r\}$ for increasingly large values of $n$ to find the point $n^*$ at which it plateaus. Given sufficient budget for experimentation, one can choose $n \approx n^*$.

%----------------------------------------------------------------------
\section{Simulation Study}\label{ssims}
%----------------------------------------------------------------------

Let $\mathcal{R} = \lambda_n [-1,1]^2$ and $\eta$ be a homogeneous Poisson point process with intensity 1. We set $\Delta = [0,\bar{\Delta}]$ for $\bar{\Delta}=0.5$ and $f$ to be the uniform distribution on $[-1,1]^2$. Conditional on $\eta$, draw $\{\varepsilon_x\}_{x\in\eta} \stackrel{iid}\sim \N(0,1)$. Potential outcomes are given by
\begin{equation*}
  Y_x(S) = \mathbb{I}_x(S) \sum_{s\in S} w_{xs}\alpha_x + \varepsilon_x, \quad \mathbb{I}_x(S) = \ind\{S \cap B(x,\bar{\Delta}) \neq \emptyset\}.
\end{equation*}

\noindent where $w_{xs} = \min\{\norm{x-s}^{-4}, 1\}$ and $\alpha_x = 1 + x_1/\lambda_n$ with $x_1$ denoting the first component of $x$. Due to the choice of $-4$ in the spatial weights $w_{xs}$, \autoref{aani} holds for any $\gamma \in (0,3)$ \citep[][Proposition 1]{leung2022rate}. The choice of $\alpha_x$ serves to generate a spatial trend. 
%Under these parameters,
%\begin{equation*}
%  \mu(s,S) = \frac{\sum_{x \in \eta \cap B(s,\bar{\Delta})} (\mathbb{I}_x(S) \sum_{z\in S} w_{xz}\alpha_x + \varepsilon_x)}{\abs{\eta \cap B(s,\bar{\Delta})}}
%\end{equation*}
%
%\noindent and
%\begin{equation*}
%  \tau(s) = \frac{\sum_{x \in \eta \cap B(s,\bar{\Delta})} w_{xs}\alpha_x}{\abs{\eta \cap B(s,\bar{\Delta})}}.
%\end{equation*}

With an exception discussed later, we choose $\beta_n$ according to \eqref{betaunder} with $\tilde\gamma \in \{1,2\}$ since $\gamma < 3$. \autoref{simresults} summarizes the results of 5000 simulation draws. Column $n^*$ reports the (average) effective sample size $\sum_{i=1}^n \ind\{i \in \M_r\}$. Column SE reports the average value of $\hat\sigma$ across draws and column CI the rate at which $\hat\theta \pm 1.96\hat\sigma$ covers $\theta_0$. All random quantities described above are redrawn each simulation, including $\eta$ and $\{\varepsilon_x\}_{x\in\eta}$. To compute $\theta_0$ for each simulation draw, we generate 5000 additional subdraws, which keep $\eta$ and $\{\varepsilon_x\}_{x\in\eta}$ fixed, in accordance with our design-based setup. 

Column $\text{SE}^*$ reports the average value of $\sigma_n = \var(\hat\theta^* \mid \X_n)^{1/2}$ across simulation draws and $\text{SE}^{**}$ reports $\bar\sigma_n = \var(\hat\theta^*)^{1/2}$ (recall \autoref{tvar}). To compute $\sigma_n$, for each simulation draw, we generate 5000 subdraws that keep $\eta$, $\{\varepsilon_x\}_{x\in\eta}$, and $\X_n$ fixed. For $\bar\sigma_n$, we redraw $\X_n$. Columns $\text{CI}^*$ ($\text{CI}^{**}$) report coverage of confidence intervals using $\text{SE}^*$ ($\text{SE}^{**}$) in place of SE.

\begin{table}[ht]
\small
\centering
\caption{Simulation Results}
\begin{threeparttable}
\resizebox{\columnwidth}{!}{%
\begin{tabular}{llrrr|rrrr|rrrr}
\toprule
$\tilde\gamma$ & $n$ & $n^*$ & $\lambda_n$ & $r$ & $\theta_0$ & $\hat\theta$ & SE & CI & $\text{SE}^*$ & $\text{CI}^*$ & $\text{SE}^{**}$ & $\text{CI}^{**}$ \\
\midrule
\multirow[m]{3}{*}{1} & 500 & 107.7 & 80 & 8.9 & 0.8162 & 0.8155 & 0.2104 & 0.9344 & 0.2076 & 0.9410 & 0.2111 & 0.9464 \\
 & 1k & 161.2 & 120 & 11.0 & 0.8159 & 0.8193 & 0.1720 & 0.9484 & 0.1673 & 0.9476 & 0.1695 & 0.9516 \\
 & 2k & 214.2 & 160 & 12.6 & 0.8157 & 0.8158 & 0.1472 & 0.9444 & 0.1427 & 0.9482 & 0.1440 & 0.9496 \\
\cmidrule(lr){1-13}
\multirow[m]{3}{*}{2} & 500 & 302.1 & 80 & 4.3 & 0.8162 & 0.8151 & 0.1299 & 0.9452 & 0.1222 & 0.9308 & 0.1281 & 0.9442 \\
 & 1k & 559.5 & 120 & 4.9 & 0.8159 & 0.8183 & 0.0967 & 0.9546 & 0.0907 & 0.9412 & 0.0947 & 0.9506 \\
 & 2k & 934.8 & 160 & 5.4 & 0.8157 & 0.8179 & 0.0758 & 0.9504 & 0.0712 & 0.9354 & 0.0737 & 0.9440 \\
\cmidrule(lr){1-13}
\multirow[m]{3}{*}{2} & 500 & 236.2 & 20 & 1.4 & 0.8165 & 0.9298 & 0.1640 & 0.9140 & 0.1406 & 0.8372 & 0.1458 & 0.8528 \\
 & 1k & 425.1 & 30 & 1.6 & 0.8158 & 0.8856 & 0.1187 & 0.9216 & 0.1054 & 0.8722 & 0.1088 & 0.8870 \\
 & 2k & 669.7 & 40 & 1.7 & 0.8158 & 0.8640 & 0.0930 & 0.9354 & 0.0840 & 0.8872 & 0.0862 & 0.9008 \\
\bottomrule
\end{tabular}}
\begin{tablenotes}[para,flushleft]
  \footnotesize 5k simulations. $n^* = \sum_{i=1}^n\ind\{i\in\M_r\}$, $\text{SE}^*=\var(\hat\theta^* \mid \X_n)^{1/2}$, $\text{SE}^{**} = \var(\hat\theta^*)^{1/2}$, and $\text{CI}^*$ ($\text{CI}^{**}$) \\ $=$ coverage of CIs using $\text{SE}^*$ ($\text{SE}^{**}$) in place of SE.
\end{tablenotes}
\end{threeparttable}
\label{simresults}
\end{table}

The results show that SE is larger than $\text{SE}^*$ and $\text{SE}^{**}$, which supports \autoref{tvar}. The first three rows report results for $\tilde\gamma=1$, resulting in a more conservative choice of $r$ that favors lower bias but higher variance. As a result, SE is larger and $n^*$ smaller compared to $\tilde\gamma=2$. Despite the relatively high value of $\gamma$, we do not see a perceptible difference in bias between the two regimes because $r$ is sizeable even in the $\tilde\gamma=2$ case.

To generate a higher bias, for the last three rows, we halve the values of $r$ from the middle three rows. Since this substantially increases the effective sample size, we also shrink $\lambda_n$ so that the values of $n^*$ are more comparable to those of the other results. Due to the increase in bias, the coverages of the ``oracle'' confidence intervals $\text{CI}^*$ and $\text{CI}^{**}$ drop substantially below the nominal level. The coverage rate of CI is noticeably better because SE is larger than $\text{SE}^*$ and $\text{SE}^{**}$, more so than the other rows. While this is partly due to the conservativeness result in \autoref{tvar}, it is likely also because the bias term $\hat\theta-\hat\theta^*$ in \eqref{decomp} contributes additional variation that the ``oracle'' SEs do not account for since they only measure the variance of $\hat\theta^*$. Since SE is constructed using $\hat\theta$, it may partly account for that variation, resulting in improved coverage.

%----------------------------------------------------------------------
\section{Conclusion}
%----------------------------------------------------------------------

This paper studies optimal experimental design for place-based interventions. We measure their effects using an ``uncontaminated'' estimand comparing responses of units near a single intervention site to the counterfactual of no intervention, averaged over all potential sites. We propose a Mat\'{e}rn thinning design that ensures sites are separated by a distance of at least $r$. This distance determines a trade-off between bias and variance, and balancing their rates is our optimality criterion.

We propose a Horvitz-Thompson estimator for $\theta_0$ and show that a rate-optimal choice of $r$ achieves the minimax rate of convergence over all potential outcomes. The rate is determined by the size of the region and the rate of decay of spillovers but not the number of hypothetical treatment sites (what we refer to as ``proto-sites''). An unlimited budget for experimentation therefore cannot generate additional power gains. Intuitively, intervention sites must be kept sufficiently separated, so a region of given size can only accommodate so many.

The Mat\'{e}rn design induces an approximate local dependence structure, which we exploit to establish asymptotic normality of the estimator. Our proof utilizes tighter Wasserstein bounds than the standard local dependence central limit theorem to obtain optimal rates. Lastly, we provide a variance estimator and conditions under which it is asymptotically conservative.

\appendix
\numberwithin{equation}{section} % include section number in equation numbering

%----------------------------------------------------------------------
\section{Proofs of Main Results}
%----------------------------------------------------------------------

%--------------------------------------
\subsection{\autoref{tclt}}
%--------------------------------------

{\bf Bias.} For $X_i \in \bm{S}$, under the event $\{D_i = 1 \cap i \in \M_{r_n}\}$,
\begin{equation}
  \abs{\mu(X_i,\bm{S}) - \mu(X_i,\{X_i\})} \leq \frac{\int_x \abs{Y_x(\bm{S}) - Y_x(\{X_i\})} \bm{1}\{\norm{X_i-x} \in \Delta\} \,\text{d}\eta}{\int_x \bm{1}\{\norm{X_i-x} \in \Delta\} \,\text{d}\eta} \leq c\,r_n^{-\gamma} \quad\text{a.s.} \label{Zbias}
\end{equation}

\noindent under Assumptions \ref{aani} and \ref{aboundY}. Thus 
\begin{equation*}
  \abs{\hat\theta - \hat\theta^*} \leq c\,r_n^{-\gamma} \cdot \frac{1}{n} \sum_{i=1}^n \frac{\bm{1}\{i \in \M_{r_n}\}}{\rho_i} \left( \frac{D_i}{p} + \frac{1-D_i}{1-p} \right).
\end{equation*}

\noindent The mean of $n^{-1} \sum_{i=1}^n \rho_i^{-1} \bm{1}\{i \in \M_{r_n}\}$ is 1, so it remains to show that its variance concentrates. Conditional on $\X_n$, the variance equals
\begin{equation}
  \frac{1}{n^2} \sum_{i=1}^n \frac{1-\rho_i}{\rho_i} + \frac{1}{n^2} \sum_{i=1}^n \sum_{j\neq i} \left( \frac{\prob(i \in \M_{r_n}, j \in \M_{r_n} \mid \X_n)}{\rho_i\rho_j} - 1 \right). \label{ve1}
\end{equation}

\noindent Let $N(i,r_n) = \abs{B(i,r_n) \cap \X_n} = \rho_i^{-1}$. The summands in the second term are identically zero if $\norm{\tilde{X}_i-\tilde{X}_j} > 2r_n$ and otherwise bounded by 2, as shown in \eqref{Se2} of \autoref{lprod}, so
\begin{equation*}
  \abs{\eqref{ve1}} \leq \frac{1}{n^2} \sum_{i=1}^n N(i,r_n) + 2\frac{1}{n^2} \sum_{i=1}^n N(i,2r_n).
\end{equation*}

\noindent Since $f$ is bounded, $\max_i N(i,c\,r_n) = O(n\beta_n^d)$ for any $c>0$, and the right-hand side is $o(1)$ since $\beta_n \rightarrow 0$. \qed

\bigskip
\noindent {\bf Normal approximation.} Since proto-sites are i.i.d., $\E[\hat\theta^* \mid \X_n] - \theta_0$ is mean zero with variance $n^{-1}\var(\tau(X_1))$, so by \autoref{aboundY}, its absolute value is $O_p(n^{-1/2})$. Suppose for the moment that 
\begin{equation}
  (\hat\theta^* - \E[\hat\theta^* \mid \X_n])/\sigma_n \dlimarrow \N(0,1).
  \label{goal}
\end{equation}

\noindent By \autoref{anondeg}, \eqref{goal} implies $\abs{\hat\theta^* - \E[\hat\theta^* \mid \X_n]} = O_p(\beta_n^{d/2})$. Since $n\beta_n^d \rightarrow \infty$, $\abs{\hat\theta^* - \E[\hat\theta^* \mid \X_n]}$ is the asymptotically dominant term, so it only remains to show \eqref{goal}.

We apply \autoref{lstein} conditioning on $\X_n$. Let $A_{ij} = \ind\{\norm{X_i-X_j} \leq 2r_n\}$, $M_i = \ind\{i \in \M_{r_n}\}$, and recall the definition of $Z_i^*$ from \eqref{Zistar}. We show that the right-hand side of \eqref{wass2} is $o(1)$ a.s., which implies convergence in distribution conditional on $\X_n$. Unconditional convergence follows from the bounded convergence theorem.

\bigskip
\noindent {\em Step 1.} Let $N(i,r_n) = \abs{B(i,r_n) \cap \X_n} = \rho_i^{-1}$, and abbreviate $N_i \equiv N(i,r_n)$. By \autoref{aboundY}, 
\begin{equation}
  \sup_n \max_i \abs{Z_i^*} \leq C(N_i M_i+1) \label{zbound}
\end{equation}

\noindent Consider the first term on the right-hand side of \eqref{wass2}. For the case $i\neq j\neq k$, 
\begin{equation*}
  \E[\abs{Z_i^*Z_j^*Z_k^*} \mid \X_n] \leq C^3\,\E[(N_iM_i+1)(N_jM_j+1)(N_kM_k+1) \mid \X_n], 
\end{equation*}

\noindent which is uniformly $O(1)$ a.s.\ by \autoref{lprod}. For the case $i=j=k$, $\E[\abs{Z_i^*Z_j^*Z_k^*} \mid \X_n] = \E[\abs{Z_i^*}^3 \mid \X_n]$, which is $O(N_i^2)$. By \autoref{ldt} and \autoref{adensity},
\begin{equation}
  \max_i N(i,c\,r_n) = O(n\beta_n^d) \quad\text{a.s.}\quad \forall\,c > 0, \label{Norder} 
\end{equation}

\noindent so the expectation is $O((n\beta_n^d)^2)$ a.s. For the case $i=j\neq k$, $\E[\abs{Z_i^*Z_j^*Z_k^*} \mid \X_n] = \E[(Z_i^*)^2 \abs{Z_k^*} \mid \X_n] \leq C \E[(N_iM_i+1)^2 (N_kM_k+1) \mid \X_n]$. By \eqref{zbound}, the leading order term is $N_i^2N_k \E[M_iM_k \mid \X_n]$, and by \autoref{lprod} and \eqref{Norder}, this is uniformly $O(n\beta_n^d)$ a.s. As a result,
\begin{multline*}
  \sum_{i=1}^n\sum_{j=1}^n\sum_{k=1}^n\E[\abs{Z_i^*Z_j^*Z_k^*} \mid \X_n] A_{ij} A_{ik} = \sum_{i=1}^n\E[\abs{(Z_i^*)^3} \mid \X_n] + \\
  3\sum_{i=1}^n\sum_{k=1}^n\E[\abs{(Z_i^*)^2Z_k^*} \mid \X_n] A_{ik} + \sum_{i=1}^n\sum_{j\neq i} \sum_{k\neq i\neq j} \E[\abs{Z_i^*Z_j^*Z_k^*} \mid \X_n] A_{ij} A_{ik} \\
  = n\cdot O((n\beta_n^d)^2) + n \cdot O((n\beta_n^d)^2) + n\cdot O((n\beta_n^d)^2) \cdot O(1) = O(n^3\beta_n^{2d}) \quad\text{a.s.}
\end{multline*}

\noindent By \autoref{anondeg}, $\sigma^2 \equiv \var(\sum_{i=1}^n Z_i^* \mid \X_n)$ is at least order $n^2\beta_n^d$, so
\begin{equation*}
  \frac{1}{\sigma^3}\sum_{i=1}^n\sum_{j=1}^n\sum_{k=1}^n\E[\abs{Z_i^*Z_j^*Z_k^*} \mid \X_n] A_{ij} A_{ik}
  = O\!\left(\frac{n^3\beta_n^{2d}}{n^3\beta_n^{3d/2}}\right) = O(\beta_n^{d/2}) \quad\text{a.s.},
\end{equation*}

\noindent which is $o(1)$ since $\beta_n\to 0$. 

\bigskip
\noindent {\em Step 2.} Consider the second term on the right-hand side of \eqref{wass2}
\begin{equation}
  \sum_{i,j,k,\ell}\cov(Z_i^*Z_j^*, Z_k^*Z_\ell^* \mid \X_n) A_{ij} A_{k\ell}. \label{covterm}
\end{equation}

\noindent There are several cases to consider. The case $i=j=k=\ell$ corresponds to
\begin{equation}
  \sum_{i=1}^n \var((Z_i^*)^2 \mid \X_n) = \sum_{i=1}^n \big(\E[(Z_i^*)^4 \mid \X_n] - \E[(Z_i^*)^2 \mid \X_n]\big). \label{alleq}
\end{equation}

\noindent By \autoref{aboundY} and \eqref{zbound},
\begin{equation}
  \E[\abs{Z_i^*}^k \mid \X_n] = O(N(i,r_n)^{k-1}), \label{lstep5}
\end{equation}

\noindent so $\eqref{alleq} = O(nN_i^3) = O_p(n^4 \beta_n^{3d})$ by \eqref{Norder}.

The case $i=j=k\neq \ell$ corresponds to
\begin{equation*}
  \sum_i \sum_{\ell\neq i} \big( \E[(Z_i^*)^3Z_\ell^* \mid \X_n] - \E[(Z_i^*)^2 \mid \X_n] \E[Z_i^*Z_\ell^* \mid \X_n] \big) A_{i\ell}.
\end{equation*}

\noindent By \eqref{lstep5} and \eqref{Norder}, $\E[(Z_i^*)^2 \mid \X_n] = O(N_i) = O(n\beta_n^d)$ uniformly in $i$. By \eqref{Se2} of \autoref{lprod}, $\E[Z_i^*Z_\ell^* \mid \X_n] \leq 2$. By \eqref{zbound}, the leading order term of $\E[(Z_i^*)^3Z_\ell^* \mid \X_n]$ is $N_i^3 N_\ell \E[M_iM_\ell \mid \X_n] \leq 2N_i^2$, which is uniformly $O((n\beta_n^d)^2)$. Therefore the right-hand side of the previous display is at most of order 
\begin{equation*}
  n \cdot n\beta_n^d \cdot ( (n\beta_n^d)^2 + n\beta_n^d ) = O(n^4 \beta_n^{3d}).
\end{equation*}

The case $i=j\neq k\neq\ell$ corresponds to
\begin{equation*}
  \sum_i \sum_{k\neq i} \sum_{\ell\neq k\neq i} \big( \E[(Z_i^*)^2 Z_k^*Z_\ell^* \mid \X_n] - \E[(Z_i^*)^2 \mid \X_n] \E[Z_k^*Z_\ell^* \mid \X_n] \big) A_{k\ell}.
\end{equation*}

\noindent The term in the parentheses is zero if $A_{ik}=0$ and $A_{i\ell}=0$, so we can upper bound this by
\begin{equation*}
  \sum_i \sum_{k\neq i} \sum_{\ell\neq k\neq i} \abs{ \E[(Z_i^*)^2 Z_k^*Z_\ell^* \mid \X_n] - \E[(Z_i^*)^2 \mid \X_n] \E[Z_k^*Z_\ell^* \mid \X_n] } (A_{ik}+A_{i\ell}) A_{k\ell}.
\end{equation*}

\noindent The leading order term of $\E[(Z_i^*)^2 Z_k^*Z_\ell^* \mid \X_n]$ is $N_i^2N_kN_\ell \E[M_iM_kM_\ell \mid \X_n]$, which is uniformly $O( n\beta_n^d)$ by \autoref{lprod} and \eqref{Norder}, so the previous term is at most of order
\begin{equation*}
  n \cdot n\beta_n^d \cdot n\beta_n^d \cdot n\beta_n^d = O(n^4 \beta_n^{3d}).
\end{equation*}

Using similar arguments, the case $i=k\neq j\neq\ell$ corresponds to
\begin{multline*}
  \sum_i \sum_{j\neq i} \sum_{\ell\neq i\neq j} \big( \E[(Z_i^*)^2Z_k^*Z_\ell^* \mid \X_n] - \E[Z_i^*Z_k^* \mid \X_n] \E[Z_i^*Z_\ell^* \mid \X_n] \big) A_{ij} A_{i\ell} \\ \leq n \cdot O(n\beta_n^d) \cdot O(n\beta_n^d) \cdot O(n\beta_n^d) = O(n^4 \beta_n^{3d}).
\end{multline*}

The case $i=j\neq k=\ell$ corresponds to
\begin{equation*}
  \sum_{i=1}^n \sum_{k\neq i} \big( \E[(Z_i^*)^2(Z_k^*)^2 \mid \X_n] - \E[(Z_i^*)^2 \mid \X_n] \E[(Z_k^*)^2 \mid \X_n] \big).
\end{equation*}

\noindent Since the term in the parentheses is zero if $A_{ik}=0$, we may multiply it by $A_{ik}$, and the result is order $n \cdot n\beta_n^d \cdot (n\beta_n^d)^2 = O(n^4\beta_n^{3d})$.

The last case is $i\neq j\neq k\neq \ell$. By \eqref{zbound},
\begin{equation*}
  \E[\abs{Z_i^*Z_j^*Z_k^*Z_\ell^*} \mid \X_n]\leq C^4 \E[(N_iM_i+1)(N_jM_j+1)(N_kM_k+1)(N_\ell M_\ell+1) \mid \X_n].
\end{equation*}

\noindent By \autoref{lprod}, $\prod_{m\in S}N_m\cdot\E[\prod_{m\in S}M_m\mid\X_n] = O(1)$ uniformly in $\abs{S}\leq 4$, so $\abs{\cov(Z_i^*Z_j^*, Z_k^*Z_\ell^* \mid \X_n)} = O(1)$ uniformly. Moreover, the covariance is zero whenever $\norm{X_a-X_b}>2r$ for all $a\in\{i,j\}$, $b\in\{k,\ell\}$, so the quadruple sum over the covariances is of asymptotic order at most
\begin{equation*}
  \sum_{i,j} A_{ij} \sum_{k,\ell} (A_{ik} + A_{i\ell} + A_{jk} + A_{j\ell}) A_{k\ell} = O(n \cdot n\beta_n^d) \cdot O((n\beta_n^d)^2) = O(n^4\beta_n^{3d}).
\end{equation*}

\noindent Putting all the cases together, by \autoref{anondeg}, 
\begin{equation*}
  \frac{2}{\pi\sigma^4}\var(W \mid \X_n) = O\left( \frac{n^4\beta_n^{3d}}{(n^2\beta_n^d)^2} \right) = o(1) \quad\text{a.s.}
\end{equation*}

\qed

%--------------------------------------
\subsection{\autoref{thm:minimax_rate}}
%--------------------------------------

Steps 2--4 of the proof are similar to \cite{faridani2023rate}, but Step 1 requires substantial additions. 

\begin{secdefinition}
  For any $\mathcal{A}\subseteq \mathbb{R}^d$, the {\em packing number} $P_{s}\left(\mathcal{A}\right)$ is the maximal number of disjoint balls of radius $s$ with centers in $A$. 
\end{secdefinition}

Fix a sequence of positive numbers $a_n \to \infty$. Suppose that there is a sequence of designs, linear estimators \eqref{eq:linear_estimators}, { such that for all } { potential outcomes satisfying Assumptions \ref{aani} and \ref{aboundY}} { and} all $\delta>0$,
\begin{equation*}
  {\prob\left(a_n|\hat{\theta}_\omega-\theta_0|> \delta\right)\to 0.}
\end{equation*}

Let $\mathcal{J}_n\subseteq \lambda_n^{-1}\text{supp}(f)$ be a maximal $a_n^{1/\gamma}/2$ packing of $\lambda_n^{-1}\text{supp}(f)$. By maximality,  $\mathcal{J}_n$ is also a $a_n^{1/\gamma}$ cover of $\lambda_n^{-1}\text{supp}(f)$. Let $k(x)\equiv \arg\min_{k\in\mathcal{J}_n} \rho(x,k)$ be the closest cover member to the location $x$. By construction $\rho(x,k(x))\leq a_n^{1/\gamma}$. 

It will be convenient to sum weights over all locations $x$ that share the same nearest cover member  $k(x)\in\mathcal{J}_n$. So, define: 
$$w_k(S)\equiv \int_x \omega_x(S)  \mathbf{1}\left\{k(x)=k\right\} \,\text{d}\eta.$$
\noindent For brevity we will suppress the dependence of $w_k$ on treatment assignment $S$ from now on.

The proof will proceed in four steps. In Step 1, we show that the ``effective sample size'' cannot be greater than $|\mathcal{J}_n|$. In Steps 2--3 we show the estimator cannot always converge faster than the square root of its effective sample size. In Step 4 we conclude the proof and solve for the minimax rate of convergence.

{\bf Step 1:}  Define the random variable $D_k$ as the indicator that at least one location within distance $a_n^{1/\gamma}$ of location $k$ was treated. Consider the potential outcomes: $Y_x(S)=a_n^{-1}D_{k(x)}$. These potential outcomes satisfy Assumptions 1--2 and do not depend on the distribution of $S$, the estimator, or $\eta$. They guarantee that any linear estimator of the form \eqref{eq:linear_estimators} can be reduced to the following weighted sum:
\begin{align*}
    \hat{\theta}_\omega &= \int_x \omega_x(S) Y_x(S) \,\text{d}\eta \\
    &= \int_x \omega_x(S) a_n^{-1}D_{k(x)} \,\text{d}\eta\\
    &= \sum_{k=1}^{|\mathcal{J}_n|} a_n^{-1}D_{k} \int_x \omega_x(S) \mathbf{1}\left\{k(x)=k\right\} \text{d}\eta\\
    &= \sum_{k=1}^{|\mathcal{J}_n|} a_n^{-1}D_{k}  w_k.
\end{align*}

Next we lower-bound $\theta_0$ under these potential outcomes. Let $\Delta^*$ denote the right endpoint of the interval $\Delta$ that defines the annulus in $\theta_0$. If $x \in \bigcup_{k\in \mathcal{J}_n} B_k( a_n^{1/\gamma}/2-\Delta^*) $, then $B_x(\Delta^*) \subseteq B_{k}(a_n^{1/\gamma}/2)$ for some $k\in \mathcal{J}_n$.  So all units within the annulus centered at $x$ share a common closest cover member $k$ and \autoref{aboundY} guarantees that this annulus has nonzero measure in $\eta$. So for such $x$ the causal effect under these potential outcomes is $\tau(x) = a_n^{-1}$. This lets us lower-bound the estimand:
\begin{equation}\label{eq:theta0_bound1}
  \theta_0 = \int_x \tau(x) \,\text{d}F(x) \geq a_n^{-1}\prob\left(X \in \bigcup_{k\in \mathcal{J}_n} B_k\left( a_n^{1/\gamma}/2-\Delta^*\right)\right).
\end{equation}

To lower-bound $\theta_0$ we need only lower-bound the probability above. Recall that $\mathcal{J}_n$ is a  $a_n^{1/\gamma}/2$-packing. So the  balls centered on the elements of $\mathcal{J}_n$ of radius $a_n^{1/\gamma}/2$ are disjoint. So:
\begin{equation*}
   \prob\left(X \in \bigcup_{k\in \mathcal{J}_n} B_k\left( a_n^{1/\gamma}/2-\Delta^*\right)\right) = \sum_{k\in \mathcal{J}_n} \prob\left( X \in B_k\left(a_n^{1/\gamma}/2-\Delta^*\right)\right).
\end{equation*}

By \autoref{lballprob}, for any sequence $s_n \to \infty$, there exist universal constants $c_1>0,c_2 \in (0,1)$ such that:
  $$ \prob\left(X \in  B_x(s_n)\right) \geq c_2\min\left\{ \frac{s_n^d}{\lambda_n^d},1\right\} \qquad \forall x \in \mathcal{R}, n>c_1.$$
\noindent We then have $\min\{ \lambda_n^{-d}(a_n^{1/\gamma}/2-\Delta^*)^d,1\} / \prob( X \in B_k(a_n^{1/\gamma}/2-\Delta^*)) = O(1)$. Thus:
\begin{equation*}
  |\mathcal{J}_n| c_2\min\left\{ \frac{(a_n^{1/\gamma}/2-\Delta^*)^d}{\lambda_n^d},1\right\} \prob\left(X \in \bigcup_{k\in \mathcal{J}_n} B_k\left( a_n^{1/\gamma}/2-\Delta^*\right)\right)^{-1} = O(1).
\end{equation*}

By \autoref{lemma:packingnumber} below, the packing number is of order $|\mathcal{J}_n| \asymp  a_n^{-d/\gamma}\lambda_n^d$.\footnote{For any sequences $\{a_n\}, \{b_n\}$, we write $a_n \asymp b_n$ to mean $a_n = O(b_n)$ and $b_n = O(a_n)$.} So, we must have:
\begin{equation*}
   \liminf_{n\to \infty}\prob\left(X \in \bigcup_{k\in \mathcal{J}_n} B_k\left( a_n^{1/\gamma}/2-\Delta^*\right)\right) >0.
\end{equation*}

\noindent Thus, combining this with (\ref{eq:theta0_bound1}), there is a constant $c_\theta>0$ such that for large enough $n$, $\theta_0 > c_\theta a_n^{-1}$.

By hypothesis, $\prob(a_n|\sum_{k=1}^{|\mathcal{J}_n|} a_n^{-1}D_kw_k - \theta_0|>\delta)\to 0$. Multiplying by $a_n$, using  the preceding lower bound on $\theta_0$, and rearranging: $\prob(\sum_{k=1}^{|\mathcal{J}_n|} D_kw_k > c_\theta/2)\to 1$. By an identical argument, $\prob(\sum_{k=1}^{|\mathcal{J}_n|} (1-D_k)w_k < -c_\theta/2)\to 1$. Combining: $ \sum_{k=1}^{|\mathcal{J}_n|}  |w_k| \geq \frac{1}{2}| \sum_{k=1}^{|\mathcal{J}_n|} D_kw_k|+\frac{1}{2}| \sum_{k=1}^{|\mathcal{J}_n|} (1-D_k)w_k|\geq c_\theta/2+o_p(1)$. So:
\begin{equation*}
    \prob\left( \sum_{k=1}^{|\mathcal{J}_n|}    |w_k| >c_\theta/4\right)\to 1.
\end{equation*}

{\bf Step 2:} Define the $|\mathcal{J}_n| \times 1$ random vectors $\mathbf{R}$ where the elements $R_k$ are  i.i.d. Rademacher random variables independent of $S$. Next we will use the Paley–Zygmund Inequality to show:
$$ \liminf_{n\to \infty}\prob\left(\sqrt{|\mathcal{J}_n|} \sum_{k=1}^{|\mathcal{J}_n|}   w_k R_k >c_\theta/8\right) >0.$$

First notice that by independence and symmetry of $R_k$, $ \sum_{k=1}^{|\mathcal{J}_n|} w_k R_k$ has the same distribution as $\sum_{k=1}^{|\mathcal{J}_n|}  |w_k|R_k$. So it suffices to work with $|w_k|$. Define $c_k \equiv |w_k|$ and define $S_n \equiv \sum_{k=1}^{|\mathcal{J}_n|}c_kR_k$. So $ \sqrt{|\mathcal{J}_n|} \sum_{k=1}^{|\mathcal{J}_n|}   w_k R_k  \stackrel{d}{=} \sqrt{|\mathcal{J}_n|}S_n$. By Step 1: $\prob( \sum_{k=1}^{|\mathcal{J}_n|} c_k > c_\theta/4) \to 1$. Conditional on the $c_k$: $\E[S_n \mid \mathbf{c}]=0$ and $\var\left(S_n \mid \mathbf{c}\right) = \sum_{k=1}^{|\mathcal{J}_n|} c_k^2$. By Cauchy-Schwarz: $\sum_{k=1}^{|\mathcal{J}_n|} c_k^2 \geq |\mathcal{J}_n| (\frac{1}{|\mathcal{J}_n|}\sum_{k=1}^{|\mathcal{J}_n|} c_k)^2 $. By the Marcinkiewicz–Zygmund Inequality: $\E\left[S_n^4\mid \mathbf{c}\right] \leq 3(\sum_{k=1}^{|\mathcal{J}_n|} c_k^2)^2$. By the  Paley–Zygmund Inequality with constant $\frac{1}{2}$:
\begin{multline}\label{eq:pz_bound}
    \prob\left(|S_n| \geq \sqrt{\frac{1}{2}\sum_{k=1}^{|\mathcal{J}_n|} c_k^2} \,\bigg|\, \mathbf{c}\right) = \prob\left(S_n^2 \geq \frac{1}{2}\mathbb{E}\left[S_n^2 \mid \mathbf{c}\right] \,\bigg|\, \mathbf{c}\right) \\
    \geq \left(1-\frac{1}{2}\right)^2 \frac{\E\left[S_n^2 \mid \mathbf{c}\right]^2}{\E\left[S_n^4 \mid \mathbf{c}\right]} 
    \geq \frac{\left(\sum_{k=1}^{|\mathcal{J}_n|} c_k^2\right)^2}{4\times 3\left(\sum_{k=1}^{|\mathcal{J}_n|} c_k^2\right)^2} \geq \frac{1}{12}.
\end{multline}

By Cauchy-Schwarz:  $(\sum_{k=1}^{|\mathcal{J}_n|} c_k^2)^{1/2} \geq |\mathcal{J}_n|^{-1/2} \sum_{k=1}^{|\mathcal{J}_n|} c_k $. By Step 1:
\begin{equation}\label{eq:cs_bound}
    \prob\left(\sqrt{\sum_{k=1}^{|\mathcal{J}_n|} c_k^2} > \frac{c_{\theta}}{4\sqrt{|\mathcal{J}_n|}} \right) \geq  \prob\left(  \sum_{k=1}^{|\mathcal{J}_n|} c_k  > \frac{c_\theta}{4}\right)\to 1
\end{equation}

By the symmetry of $S_n$, for any $\delta \in (0,c_\theta / (4\sqrt{2}))$: $  \prob(\sqrt{|\mathcal{J}_n|} S_n \geq  \delta ) \geq 0.5 \prob(|S_n|\geq \delta / \sqrt{|\mathcal{J}_n|})$. Then by (\ref{eq:pz_bound}) and (\ref{eq:cs_bound}):
\begin{multline*}
    \prob\left(|S_n|\geq \delta / \sqrt{|\mathcal{J}_n|}\right) \geq \prob\left(|S_n| \geq \sqrt{\frac{1}{2}\sum_{k=1}^{|\mathcal{J}_n|} c_k^2} \cap \sqrt{\sum_{k=1}^{|\mathcal{J}_n|} c_k^2} > \sqrt{2}\delta /\sqrt{|\mathcal{J}_n|}\right)  \\ \geq   \prob\left(|S_n| \geq \sqrt{\frac{1}{2}\sum_{k=1}^{|\mathcal{J}_n|} c_k^2}\mid \sqrt{\sum_{k=1}^{|\mathcal{J}_n|} c_k^2} > \sqrt{2}\delta /\sqrt{|\mathcal{J}_n|}\right)  \prob\left(\sqrt{\sum_{k=1}^{|\mathcal{J}_n|} c_k^2} > \sqrt{2}\delta / \sqrt{|\mathcal{J}_n|}\right) \\ \geq \frac{1}{12}+o(1).
\end{multline*}

So $\liminf_{n\to \infty}\prob( \sqrt{|\mathcal{J}_n|} S_n  \geq  \delta ) > 0$. Now substitute $ \sqrt{|\mathcal{J}_n|} \sum_{k=1}^{|\mathcal{J}_n|}   w_k R_k  = \sqrt{|\mathcal{J}_n|} S_n$. So we have shown:
\begin{equation}\label{eq:rademacher}
    \liminf_{n\to \infty}\prob\left( \sqrt{|\mathcal{J}_n|} \sum_{k=1}^{|\mathcal{J}_n|}   w_kR_k >\frac{c_\theta}{8}\right)=\liminf_{n\to \infty}\prob\left( \sqrt{|\mathcal{J}_n|} \sum_{k=1}^{|\mathcal{J}_n|}   |w_k|R_k >\frac{c_\theta}{8}\right) >0.
\end{equation}

{\bf Step 3:} We now show that there exists another sequence of potential outcomes (that may depend on the design) such that: $$ {\liminf_{n\to \infty}\prob\left(\sqrt{|\mathcal{J}_n|}\left|\hat{\theta}_\omega-\theta_0\right|> \frac{c_\theta}{8}\right) >0}$$ 

To do this, notice that since (\ref{eq:rademacher}) holds over randomly selected $\mathbf{R}$, it also holds conditional on some non-stochastic realized vectors $\mathbf{r}\in \{-1,1\}^{|\mathcal{J}_n|}$:
$$  \liminf_{n\to \infty}\max_{\mathbf{r}\in \{-1,1\}^{|\mathcal{J}_n|}}\prob\left( \sqrt{|\mathcal{J}_n|}\sum_{k=1}^{|\mathcal{J}_n|}   w_k r_k >\frac{c_\theta}{8}\right) >0.$$

Now consider the potential outcomes $Y_x(S)=r_{k(x)}$ where the vector $\mathbf{r}$ for each $n$ is the maximizer in the limit above. Since these do not depend on treatment at all, $\theta_0=0$. These satisfy Assumptions \ref{aani} and \ref{aboundY}. Under these potential outcomes: 
$$ \liminf_{n\to \infty}\prob\left(\sqrt{|\mathcal{J}_n|}\left|\hat{\theta}_\omega-\theta_0\right|> \frac{c_\theta}{8}\right)=\liminf_{n\to \infty}\max_{\mathbf{r}\in \{-1,1\}^{|\mathcal{J}_n|}}\prob\left( \sqrt{|\mathcal{J}_n|} \sum_{k=1}^{|\mathcal{J}_n|}   w_k r_k >\frac{c_\theta}{8}\right) >0.$$

{\bf Step 4:} We have two guarantees, one by hypothesis and the other was established in Step 3. {First,}
\begin{equation*}
  {\prob\left(a_n\left|\hat{\theta}_\omega-\theta_0\right|>\frac{c_\theta}{8}\right) \to 0}
\end{equation*}

\noindent {for any sequence of potential outcomes satisfying Assumptions \ref{aani} and \ref{aboundY}. Second, there exists a sequence of potential outcomes satisfying the assumptions (that might depend on the design and estimator)  such that:}
\begin{equation*}
 {\liminf_{n\to \infty}\prob\left(\sqrt{|\mathcal{J}_n|}\left|\hat{\theta}_\omega-\theta_0\right|> \frac{c_\theta}{8}\right)>0.}
\end{equation*}

\noindent For both to be true, it is necessary that: $a_n/\sqrt{|\mathcal{J}_n|} \to 0$. By Lemma \ref{lemma:packingnumber},  $a_n = O(\sqrt{|\mathcal{J}_n|}) = O(\lambda_n^{d/2}a_n^{-d/(2\gamma)})$. Solving:
$$ a_n = O\big( \lambda_n^{\frac{d}{2+d/\gamma}} \big).$$
 
\qed

\begin{seclemma}\label{lemma:packingnumber}
  Under \autoref{adensity}, for all sequences $s_n \to \infty$, the packing number satisfies
  $$ P_{s_n}\left(\lambda_n^{-1}\text{supp}(f)\right) \asymp  \max\left\{\frac{\lambda_n^d}{s_n^d},1 \right\}.$$
\end{seclemma}
\begin{proof}
    Consider three cases. In the case where $\lambda_n/s_n \to 0 $, then the packing number is one and the claim is satisfied. If $\lambda_n/s_n$ is bounded above, then both the right and left hand sides are bounded and neither converges to zero and the claim is satisfied. The rest of the proof covers the case where $\lambda_n/s_n \to \infty$.

    By isomorphism, it will be sufficient to prove that $P_{t_n}\left(\text{supp}(f)\right) \asymp  \max\{t_n^{-d},1\} $ for all sequences $t_n \to 0 $.

    First we show that  $P_{t_n}\left(\text{supp}(f)\right) = O(\max\{t_n^{-d},1\})$. To see why, assume for the sake of contradiction that $P_{t_n}\left(\text{supp}(f)\right)t_n^d \to \infty$. The $t_n$ packing number of any fixed ball scales with $t_n^{-d}$. So if $P_{t_n}\left(\text{supp}(f)\right)$ grows faster than $t_n^{-d}$, then  there is no ball of any radius that can contain $\text{supp}(f)$, which contradicts the boundedness of Assumption \ref{adensity}. Therefore, $P_{t_n}\left(\text{supp}(f)\right) = O(\max\{t_n^{-d},1\})$.

    Next we show that $\max\{t_n^{-d},1 \} = O(P_{t_n}(\text{supp}(f)))$. To see why, notice that any $t_n$-packing is a $2t_n$-cover. So $\text{Vol}(\text{supp}(f)) \leq 2^d t_n^d P_{t_n}\left(\text{supp}(f)\right)$. Assume for sake of contradiction that $ P_{t_n}\left(\text{supp}(f)\right) = o(t_n^d)$. Then $2^d t_n^d P_{t_n}\left(\text{supp}(f)\right) \to 0$ and $\text{Vol}(\text{supp}(f))$ would have to equal zero, which contradicts the assumption of the lemma. 

    We have shown that for all $t_n\to 0$ we have $\max\{t_n^{-d},1 \} = O(P_{t_n}(\text{supp}(f)))$ and $P_{t_n}(\text{supp}(f))) = O(\max\{s_n^{-d},1\})$. So, $P_{t_n}(\text{supp}(f)) \asymp  \max\{t_n^{-d},1\} $ and thus $P_{t_n}(\lambda_n^{-1}\text{supp}(f)) \asymp  \max\{s_n^{-d}\lambda_n^d,1\} $.
    
\end{proof}

%--------------------------------------
\subsection{\autoref{tlimvar}}
%--------------------------------------

Under \autoref{astat}, redefine $\tau(x) \equiv \mu^1(x) - \mu^0(x)$, so the argument of $\tau$ is $\tilde{X}_i$. By \eqref{s2e}, $\var(\sum_{i=1}^n Z_i^* \mid \X_n) = T_1 + T_2 + T_3$ where
\begin{align*}
  T_1 &= \frac{1}{p(1-p)} \sum_{i=1}^n N(i,r_n)\big( (1-p)\mu^1(\tilde{X}_i) + p\mu^0(\tilde{X}_i) \big)^2, \\
  T_2 &= \sum_{\substack{i < j \\ \norm{\tilde{X}_i-\tilde{X}_j}\leq \beta_n}} (\tau(\tilde{X}_i) - \tau(\tilde{X}_j))^2, \\
  T_3 &= \sum_{\substack{i \neq j \\ \beta_n < \norm{\tilde{X}_i-\tilde{X}_j} \leq 2\beta_n}} \frac{m_{ij}}{N(i,r_n) + N(j,r_n) - m_{ij}} \,\tau(\tilde{X}_i) \tau(\tilde{X}_j).
\end{align*}

\noindent {\bf Term $T_1$.} By \autoref{ldt}, $N(i,r_n) = n\omega_d\beta_n^d f(\tilde X_i)(1+o(1))$ a.s.\ uniformly in $i$, so 
\begin{multline*}
  \sum_{i=1}^n N(i,r_n) \big( (1-p)\mu^1(\tilde{X}_i) + p\mu^0(\tilde{X}_i) \big)^2 \\ = n\omega_d\beta_n^d \sum_{i=1}^n f(\tilde{X}_i) \big( (1-p)\mu^1(\tilde{X}_i) + p\mu^0(\tilde{X}_i) \big)^2 + o(n^2\beta_n^d) \quad \text{a.s.}
\end{multline*}

\noindent using \autoref{adensity}. Applying the strong law of large numbers to the first term on the right,
\begin{multline*}
  (n^2\beta_n^d)^{-1} \frac{1}{p(1-p)} \sum_{i=1}^n N(i,r_n) \big( (1-p)\mu^1(\tilde{X}_i) + p\mu^0(\tilde{X}_i) \big)^2 \\ \stackrel{a.s.}\longrightarrow \frac{\omega_d}{p(1-p)} \int_x \big( (1-p)\mu^1(x) + p\mu^0(x) \big)^2 f(x)^2 \,\text{d}x.
\end{multline*}

\noindent {\bf Term $T_2$.} By \autoref{astat}, $\tau$ is continuous on the compact set $\mathrm{supp}(f)$, hence uniformly continuous. For any $\varepsilon > 0$ and $n$ sufficiently large, every pair with $\abs{\tilde X_i - \tilde X_j} \leq \beta_n$ satisfies $\abs{\tau(\tilde{X}_i) - \tau(\tilde{X}_j)} < \varepsilon$. By \eqref{Norder},
\begin{equation*}
  \sum_{\substack{i < j \\ \norm{\tilde{X}_i-\tilde{X}_j}\leq \beta_n}} (\tau(\tilde{X}_i) - \tau(\tilde{X}_j))^2 \leq \varepsilon^2 \sum_{i=1}^n N(i,r_n) = O(\varepsilon^2 n^2\beta_n^d) \quad \text{a.s.}
\end{equation*}

\noindent Since $\varepsilon > 0$ is arbitrary, $T_2 = o(n^2\beta_n^d)$ a.s.

\bigskip
\noindent {\bf Term $T_3$.} Let $\delta_{ij} = \norm{\tilde X_i - \tilde X_j}$. By \eqref{intersectunif}, $m_{ij} = n\beta_n^d v_d(\delta_{ij}/\beta_n) f(\tilde X_i)(1+o(1))$ a.s.\ uniformly over $i,j$ such that $\delta_{ij} \in (\beta_n, 2\beta_n)$. Together with \autoref{ldt},
\begin{equation*}
  \max_{i,j\colon \delta_{ij} \in (\beta_n, 2\beta_n)} \bigg| \frac{m_{ij}}{N(i,r_n)+N(j,r_n)-m_{ij}} - \frac{v_d(\delta_{ij}/\beta_n)}{2\omega_d - v_d(\delta_{ij}/\beta_n)} \bigg| \stackrel{a.s.}\longrightarrow 0.
\end{equation*}

\noindent Then
\begin{equation*}
  T_3 = \sum_{\substack{i \neq j \\ \beta_n < \norm{\tilde{X}_i-\tilde{X}_j} \leq 2\beta_n}} \frac{v_d(\delta_{ij}/\beta_n)}{2\omega_d-v_d(\delta_{ij}/\beta_n)} \,\tau(\tilde X_i)\tau(\tilde X_j) + o(n^2\beta_n^d) \quad\text{a.s.}
\end{equation*}

\noindent Call the first term on the right-hand side $U_n$, which is proportional to a U-statistic with uniformly bounded kernel by \autoref{aboundY}. By a Bernstein inequality \citep[e.g.][Theorem 1.6]{ai2022hoeffding}, there exists $C>0$ such that for any $n$,
\begin{equation*}
  \prob\left( \abs{(n^2\beta_n^d)^{-1} (U_n-\E[U_n])} > t \right) \leq 2\,\text{exp}\left\{ - C\frac{n (t\beta_n^d)^2}{s_n^2 + t\beta_n^d} \right\}
\end{equation*} % take \varepsilon = t\beta_n^d

\noindent where $s_n^2$ is the variance of the U-statistic summand, which is bounded above by $\prob(\norm{\tilde{X}_i-\tilde{X}_j} \leq 2\beta_n)$ by \autoref{aboundY}. The latter is $O(\beta_n^d)$ by \autoref{ldt}, so the right-hand side scales like $\text{exp}\{-Cn\beta_n^d\}$. 

By assumption, $n\beta_n^d / \log(n) \rightarrow c$ sufficiently large, so taking $c > C^{-1}$, $\text{exp}\{-Cn\beta_n^d\}$ is summable over $n$. By Borel-Cantelli, $(n^2\beta_n^d)^{-1} U_n$ converges almost surely to the limit of its expectation
\begin{equation*}
  \beta_n^{-d} \int_{x \in \R^d} \int_{x'\colon \norm{x-x'} \in (\beta_n,2\beta_n)} \frac{v_d(\norm{x-x'}/\beta_n)}{2\omega_d-v_d(\norm{x-x'}/\beta_n)} \,\tau(x)\tau(x') f(x) f(x') \,\text{d}x \,\text{d}x'.
\end{equation*}

\noindent By a change of variables $z \equiv (x'-x)/\beta_n$, continuity of $\tau$ and $f$, and the bounded convergence theorem, the limit of the previous display is
\begin{equation*}
  \int_{\norm{z}\in (1,2)} \frac{v_d(\norm{z})}{2\omega_d-v_d(\norm{z})} \,\text{d}z \int_x \tau(x)^2 f(x)^2 \,\text{d}x.
\end{equation*}

\noindent Thus $T_3 = n^2\beta_n^d(U^* + o(1))$ a.s.

\qed

%--------------------------------------
\subsection{\autoref{tvar}}
%--------------------------------------

The first two steps of the proof establish $(\hat\sigma^2 - \mathcal{B}_n)/\sigma_n^2 \plimarrow 1$. The third step derives the probability limit of $\mathcal{B}_n$. The fourth step proves $(\hat\sigma^2 - \mathcal{B}_n)/\bar{\sigma}_n^2 \plimarrow 1$.

{\bf Step 1.} Abbreviate $K_{ij} = \ind\{ \norm{\tilde{X}_i-\tilde{X}_j} \leq 2\beta_n \}$. Let
\begin{equation*}
  W_i = \frac{\bm{1}\{i \in \M_{r_n}\}}{\rho_i} \left( \frac{\mu(X_i,\{X_i\})D_i}{p} - \frac{\mu(X_i,\{\emptyset\})(1-D_i)}{1-p} \right)
\end{equation*}

\noindent and $\bar{W} = n^{-1} \sum_{i=1}^n W_i$. By \eqref{Zbias}, \autoref{aboundY}, \eqref{Norder}, and \autoref{tclt},
\begin{align*}
  \beta_n^{-d} \hat\sigma^2 &= \frac{1}{n^2\beta_n^d} \sum_{i=1}^n \sum_{j=1}^n K_{ij} (W_i-\bar{W}) (W_j-\bar{W}) + O_p(r_n^{-\gamma}) \\
			    &= \frac{1}{n^2\beta_n^d} \sum_{i=1}^n \sum_{j=1}^n K_{ij} (W_i-\E[\bar{W} \mid \X_n]) (W_j-\E[\bar{W} \mid \X_n]) + O_p(r_n^{-\gamma} + \beta_n^{d/2}).
\end{align*}

\noindent The remainder term is $o_p(1)$. Further decompose
\begin{multline*}
  \frac{1}{n^2\beta_n^d} \sum_{i=1}^n \sum_{j=1}^n K_{ij} (W_i-\E[\bar{W} \mid \X_n]) (W_j-\E[\bar{W} \mid \X_n]) = \\ \frac{1}{n^2\beta_n^d} \sum_{i=1}^n \sum_{j=1}^n K_{ij} (W_i-\E[W_i \mid X_i]) (W_j-\E[W_j \mid X_j]) \\ + \frac{1}{n^2\beta_n^d} \sum_{i=1}^n \sum_{j=1}^n K_{ij} (\E[W_i \mid X_i]-\E[\bar{W} \mid \X_n]) (\E[W_j \mid X_j]-\E[\bar{W} \mid \X_n]) \equiv [I] + [II].
\end{multline*}

\noindent {\bf Step 2.} We show that $[I]$ is consistent for $\beta_n^{-d} \sigma_n^2$. Recall the definition of $Z_i^*$ from \eqref{Zistar}, and notice
\begin{equation*}
  [I] = \frac{1}{n^2\beta_n^d} \sum_{i=1}^n \sum_{j=1}^n K_{ij} Z_i^* Z_j^*,
\end{equation*}

\noindent so $\E[[I] \mid \X_n] = \beta_n^{-d} \sigma_n^2$. Its conditional variance is
\begin{multline*}
  \frac{1}{n^4\beta_n^{2d}} \sum_{i=1}^n \sum_{j=1}^n \sum_{k=1}^n \sum_{\ell=1}^n K_{ij}K_{k\ell} \cov(Z_i^*Z_j^*, Z_k^*Z_\ell^* \mid \X_n) \\ 
  \leq \frac{1}{n^4\beta_n^{2d}} \sum_{i=1}^n \sum_{j=1}^n \sum_{k=1}^n \sum_{\ell=1}^n K_{ij}K_{k\ell}(K_{ik}+K_{i\ell}+K_{jk}+K_{j\ell}) \cov(Z_i^*Z_j^*, Z_k^*Z_\ell^* \mid \X_n)
\end{multline*}

\noindent since $Z_i^*Z_j^* \indep Z_k^*Z_\ell^* \mid \X_n$ if $K_{ik}+K_{i\ell}+K_{jk}+K_{j\ell}=0$. At shown in step 2 of the proof of \autoref{tclt}, the conditional covariance is $O(1)$ in absolute value a.s., so by \eqref{zbound}, the upper bound above is a.s.\ of order $n^{-4}\beta_n^{-2d} n (n\beta_n^d)^3 = \beta_n^d = o(1)$.

\bigskip
\noindent {\bf Step 3.} We derive the probability limit of 
\begin{equation*}
  [II] = \frac{1}{n^2\beta_n^d}\sum_{i=1}^n (\tau(X_i)-\bar\tau)^2 + \beta_n^{-d} \mathcal{B}_n
\end{equation*}

\noindent The first term on the right is $O((n\beta_n^d)^{-1}) = o(1)$ a.s.\ by the strong law. The second term is
\begin{equation}
  \frac{1}{n^2 \beta_n^d} \sum_{i=1}^n \sum_{j\neq i} K_{ij} (\tau(X_i) - \theta_0)(\tau(X_j) - \theta_0). \label{offd}
\end{equation}

\noindent by \autoref{aboundY} and proto-sites being i.i.d. Its expectation equals
\begin{equation*}
  \frac{n-1}{n\beta_n^d}\int_x \int_y \bm{1}\left\{\norm{x-y} \leq 2\beta_n\right\} (\tau(\lambda_n x) - \theta_0) (\tau(\lambda_n y) - \theta_0) f(x) f(y) \,\text{d}x \,\text{d}y.
\end{equation*}

\noindent By a change of variables $y = x + \beta_n u$, this equals
\begin{equation*}
  \frac{n-1}{n}\int_x \int_{\norm{u} \leq 2} (\tau(\lambda_n x) - \theta_0) (\tau(\lambda_n(x+\beta_n u)) - \theta_0) f(x) f(x+\beta_n u) \,\text{d}u \,\text{d}x.
\end{equation*}

\noindent Since $f$ is continuous on its support, $\tau(\lambda_n\tilde{X}_1) \stackrel{a.s.}\longrightarrow \tau_\infty(\tilde{X}_1)$, and $(\lambda_n x) / (\lambda_n (x+\beta_n u)) \rightarrow 1$, this has the desired limit \eqref{Bconv} by bounded convergence.

It remains to bound the variance of \eqref{offd}, which is
\begin{equation*}
  \frac{1}{n^4\beta_n^{2d}} \sum_{i=1}^n \sum_{j=1}^n \sum_{k=1}^n \sum_{\ell=1}^n \cov( K_{ij} (\tau(X_i) - \theta_0)(\tau(X_j) - \theta_0), K_{k\ell} (\tau(X_k) - \theta_0)(\tau(X_\ell) - \theta_0) ).
\end{equation*}

\noindent Since $\{X_i\}_{i=1}^n$ is i.i.d., this is zero when $i\neq j\neq k\neq \ell$, $i\neq j\neq k=\ell$, and $i=j \neq k=\ell$. The contribution of case $i=j=k=\ell$ only involves one summation and is therefore $O(n^{-4}\beta_n^{2d} n) = o(1)$ by \autoref{aboundY}. The contribution of case $i\neq j=k=\ell$ only involves two summations and is therefore $O(n^{-4}\beta_n^{2d} n^2) = o(1)$. Finally, the contribution of case $i\neq j=k\neq \ell$ is
\begin{multline*}
  \frac{1}{n^4\beta_n^{2d}} \sum_{i=1}^n \sum_{j=1}^n \sum_{\ell=1}^n \big( \E[K_{ij}K_{j\ell}(\tau(X_i) - \theta_0)(\tau(X_j) - \theta_0)^2(\tau(X_\ell) - \theta_0)] \\ - \E[K_{ij}(\tau(X_i) - \theta_0)(\tau(X_j) - \theta_0)] \E[K_{j\ell}(\tau(X_j) - \theta_0)(\tau(X_\ell) - \theta_0)] \big).
\end{multline*}

\noindent By \autoref{ldt} and \autoref{adensity}, this is $O(n^{-1}\beta_n^{-2d} \beta_n^{2d}) = o(1)$. 

\bigskip
\noindent {\bf Step 4.} By the law of total variance, $\bar{\sigma}_n^2 = \E[\sigma_n^2] + \var(n^{-1} \sum_{i=1}^n \tau(X_i))$. Step 2 showed that $\beta_n^{-d}\hat\sigma^2 = \beta_n^{-d}\sigma_n^2 + o_p(1)$. In the proof of \autoref{tlimvar}, it is shown that $\beta_n^{-d}\sigma_n^2 = n^{-2}\beta_n^{-d}(T_1 + T_2 + T_3)$, and both $n^{-2}\beta_n^{-d}T_1$ and $n^{-2}\beta_n^{-d}T_3$ are consistent for their expectations. Provided $n^{-2}\beta_n^{-d}T_2$ has the same property, $\beta_n^{-d}\sigma_n^2 = \beta_n^{-d}\E[\sigma_n^2] + o_p(1)$.

Term $n^{-2}\beta_n^{-d}T_2$ is proportional to a U-statistic. By Theorem 1.6 of \cite{ai2022hoeffding},
there exists $C>0$ such that for any $n$,
\begin{equation*}
  \prob\left( \abs{(n^2\beta_n^d)^{-1} (U_n-\E[U_n])} > t \right) \leq 2\,\text{exp}\left\{ - C\frac{n (t\beta_n^d)^2}{s_n^2 + t\beta_n^d} \right\}
\end{equation*} % take \varepsilon = t\beta_n^d

\noindent where $s_n^2$ is the variance of the U-statistic summand, which is bounded above by $\prob(\norm{\tilde{X}_i-\tilde{X}_j} \leq 2\beta_n)$ by \autoref{aboundY}. The latter is $O(\beta_n^d)$ by \autoref{ldt}, so the right-hand side scales like $\text{exp}\{-Cn\beta_n^d\}$. By assumption, $n\beta_n^d / \log(n) \rightarrow c$ sufficiently large, so taking $c > C^{-1}$, $\text{exp}\{-Cn\beta_n^d\}$ is summable over $n$. By Borel-Cantelli, $(n^2\beta_n^d)^{-1} U_n$ converges almost surely to the limit of its expectation.

The remaining part is
\begin{equation*}
  \beta_n^{-d} \var\bigg(\frac{1}{n} \sum_{i=1}^n \tau(X_i)\bigg) = \frac{1}{n^2\beta_n^d} \sum_{i=1}^n \var(\tau(X_i)) = o(1)
\end{equation*}

\noindent by \autoref{aboundY}. Therefore, $\beta_n^{-d}\hat\sigma^2 = \sigma_n^2 + \mathcal{B}_n + o_p(1)$ by step 2, which equals $\E[\sigma_n^2] + \mathcal{B}_n + o_p(1)= \bar\sigma_n^2 + \mathcal{B}_n + o_p(1)$.

\qed

%----------------------------------------------------------------------
\section{Auxiliary Lemmas}
%----------------------------------------------------------------------

Throughout we let $N(i,r) = \abs{B(i,r) \cap \X_n}$.

%--------------------------------------
\subsection{Local Volumes}
%--------------------------------------

\begin{seclemma}\label{ldt}
  Let $\omega(x) = \text{Vol}(B(x,1) \cap \text{supp}(f))$ and $p_i = \int_{B(\tilde X_i, \beta_n)} f(z) \,\text{d}z$. If $\beta_n \to 0$, then
  \begin{equation}
    \max_i \left| \frac{N(i,r_n)}{n} - p_i \right| = o(1) \quad\text{a.s.}
    \label{hbc}
  \end{equation}

  \noindent and
  \begin{equation}
    \frac{p_i}{\beta_n^d} \longrightarrow \omega(\tilde{X}_i) f(\tilde X_i) > 0 \quad\text{a.s.} \label{lebdt}
  \end{equation}

  \noindent Furthermore, if $f$ is continuous on its support, then \eqref{lebdt} holds uniformly over $i$. 
\end{seclemma}
\begin{proof}
  Since $N(i,r_n) \sim \mathrm{Bin}(n-1,p_i) + 1$, by the union bound and Hoeffding's inequality, 
  \begin{equation*}
    \prob\left( \max_i \abs{N(i,r_n) - (n-1)p_i - 1} \geq \varepsilon n \right) \leq 2n\cdot \text{exp}\left\{ -\frac{2\varepsilon^2 n^2}{n-1} \right\}.
  \end{equation*}

  \noindent Summing the right-hand side from $n=1$ to $\infty$, the result is finite, so we obtain \eqref{hbc} from Borel-Cantelli.

  Since $f$ is a density function, it is $L_1$ integrable, and $f(\tilde X_i)>0$ a.s. Given $\beta_n\to 0$,
  \begin{equation*}
    \frac{p_i}{\beta_n^d} = \beta_n^{-d} \int_{B(\tilde{X}_i,\beta_n)\cap\mathrm{supp}(f)} f(z) \,\text{d}z \stackrel{a.s.}\longrightarrow \omega(\tilde{X}_i) f(\tilde{X}_i)
  \end{equation*}

  \noindent by the Lebesgue differentiation theorem.

  Finally, since $f$ is continuous on its support, and its support is compact, being a subset of the bounded region, it follows that $f$ is uniformly continuous on its support, so
  \begin{multline}
    \max_i \bigg| \frac{1}{\omega(\tilde{X}_i) \beta_n^d} p_i - f(\tilde{X}_i) \bigg| \leq \sup_{x\in\text{supp}(f)} \bigg| \frac{1}{\omega(x) \beta_n^d} \int_{B(x, \beta_n) \cap \text{supp}(f)} f(z) \,\text{d}z - f(x) \bigg| \\
    \leq \sup_{x\in\text{supp}(f)} \sup_{z \in B(x, \beta_n) \cap \text{supp}(f)} \big| f(z) - f(x) \big| = o(1) \quad\text{a.s.} \label{ldtunif} 
  \end{multline}
\end{proof}

\begin{seclemma}\label{lNm}
  Under \autoref{adensity}, if $\beta_n \to 0$, then 
  \begin{equation*}
    \max_{i,j\colon \norm{X_i-X_j} > r_n} \frac{N(i,r_n)}{N(i,r_n) - \abs{B(X_i,r_n) \cap B(X_j,r_n) \cap \X_n}} = O(1) \quad\text{a.s.}
  \end{equation*}
\end{seclemma}
\begin{proof}
  Let $m_{ij} = \abs{B(X_i,r_n) \cap B(X_j,r_n) \cap \X_n}$ and
  \begin{equation*}
    p_{ij} = \int_{B(\tilde X_i, \beta_n) \cap B(\tilde X_j, \beta_n)} f(z)\,\text{d}z.
  \end{equation*}

  \noindent Note that $N(i,r_n) - m_{ij}$ is always strictly positive given $\norm{X_i-X_j} > r_n$. By the same argument as \eqref{hbc}, $\max_{i,j} \abs{m_{ij}/n - p_{ij}} = o(1)$ a.s. Using \eqref{hbc},
  \begin{equation*}
    \max_{i,j} \frac{N(i,r_n)}{N(i,r_n) - m_{ij}} = \max_{i,j} \frac{1}{1 - p_{ij} / p_i} + o(1) \quad\text{a.s.}
  \end{equation*}

  \noindent for $p_i$ in \eqref{hbc}. 

  It remains to show that $p_{ij} / p_i$ is uniformly bounded below one. If $\norm{\tilde{X}_i-\tilde{X}_j} \geq 2\beta_n$, then $p_{ij}=0$, so it remains to consider the case $\norm{\tilde{X}_i-\tilde{X}_j} \in (\beta_n,2\beta_n)$.

  By \autoref{adensity}, $\text{supp}(f)$ has a Lipschitz boundary, so by 4.11 of \cite{adams2003sobolev}, it satisfies the cone condition (\autoref{dcc}) for some finite cone $C(\rho, v, \kappa)$. Then for $n$ sufficiently large and some rotation $\bm{R}_x\colon \R^d\rightarrow\R^d$, $B(x,\beta_n) \cap \text{supp}(f)$ contains the cone $x + \bm{R}_x(C(\beta_n v, \kappa))$ whose volume is a fixed fraction $c_0$ of $\text{Vol}(B(x,\beta_n))$ for $c_0 \in (0,1]$ independent of $x,n$. Hence
  \begin{equation*}
    0 < c_0 \leq c_x \equiv \frac{\text{Vol}(B(x, \beta_n) \cap \text{supp}(f))}{\text{Vol}(B(x, \beta_n))}.
  \end{equation*}

  \noindent Also let
  \begin{equation*}
    c_{x,x'} = \frac{\text{Vol}(B(x, \beta_n) \cap B(x', \beta_n) \cap \text{supp}(f))}{\text{Vol}(B(x, \beta_n) \cap B(x', \beta_n))},
  \end{equation*}

  \noindent and note that $c_{x,x'} \leq c_x$.

  Define $v(x,y) = \text{Vol}(B(0,1)\cap B(y-x,1))$, so $v(x,y) \in (0,\omega_d)$ when $\norm{x-y} \in (0,2)$ and $v(x,y) = 0$ when $\norm{x-y} \geq 2$. By the change of variables $z \mapsto (z-x)/\beta_n$,
  \begin{equation*}
    \text{Vol}(B(x,\beta_n) \cap B(x',\beta_n)) = \beta_n^d\, v(0,(x'-x)/\beta_n).
  \end{equation*}

  \noindent Since $f$ is continuous on its support by \autoref{adensity}, \eqref{ldtunif} holds, and by a similar argument,
  \begin{multline}
    \max_{i,j\colon \norm{\tilde{X}_i - \tilde{X}_j} \in (\beta_n,2\beta_n)} \bigg| \frac{1}{c_{x,x'} \beta_n^d v(0,(\tilde{X}_j-\tilde{X}_i)/\beta_n)} p_{ij} - f(\tilde{X}_i) \bigg| \\ \leq \sup_{\substack{x,x'\in\text{supp}(f)\colon \\ \norm{x-x'}\in (\beta_n,2\beta_n)}} \bigg| \frac{1}{c_{x,x'} \beta_n^d v(0,(x'-x)/\beta_n)} \int_{B(x, \beta_n) \cap B(x', \beta_n) \cap \text{supp}(f)} f(z) \,\text{d}z - f(x) \bigg| \\
    \leq \sup_{\substack{x,x'\in\text{supp}(f)\colon \\ \norm{x-x'}\in (\beta_n,2\beta_n)}} \sup_{z \in B(x, \beta_n) \cap B(x', \beta_n) \cap \text{supp}(f)} \big| f(z) - f(x) \big| = o(1) \label{intersectunif}
  \end{multline}

  \noindent Therefore,
  \begin{equation*}
    \frac{p_{ij}}{p_i} = \frac{c_{x,x'} \beta_n^d v(0,(\tilde{X}_j-\tilde{X}_i)/\beta_n) f(\tilde{X}_i)}{c_x \omega_d \beta_n^d f(\tilde{X}_i)} + o(1) \leq \frac{v(0,(\tilde{X}_j-\tilde{X}_i)/\beta_n)}{\omega_d} + o(1)
  \end{equation*}

  \noindent a.s.\ uniformly over $i,j$ such that $\norm{\tilde{X}_i - \tilde{X}_j} \in (\beta_n, 2\beta_n)$. Under the latter event, $\norm{(\tilde{X}_j-\tilde{X}_i)/\beta_n} \in (1,2)$, so $v(0,(\tilde{X}_j-\tilde{X}_i)/\beta_n) < \omega_d$, meaning $p_{ij}/p_i$ is asymptotically bounded strictly below one, uniformly over such $i,j$, as desired.
\end{proof}

%--------------------------------------
\subsection{Moments}
%--------------------------------------

\begin{seclemma}\label{lprod}
  Let $M_i = \ind\{i \in \M_{r_n}\}$ and $N_i = N(i,r_n)$. Then $\max\{ \prod_{s \in S} N_s \E[\prod_{s\in S} M_s \mid \X_n]\colon S \subseteq \{1,\ldots,n\}, \abs{S}\leq 4\} = O(1)$ a.s.
\end{seclemma}
\begin{proof}
  {\bf Case $\abs{S}=1$.} Take $S = \{i\}$. Then $N_i \E[M_i \mid \X_n]= N_i \cdot N_i^{-1}=1$.

  \bigskip
  \noindent {\bf Case $\abs{S}=2$.} Take $S=\{i,j\}$. If $\norm{X_i-X_j}\leq r_n$, then $M_i M_j=0$ a.s. If $\norm{X_i-X_j} > 2r_n$, then $M_i \indep M_j \mid \X_n$, so $N_i N_j\,\E[M_i M_j \mid \X_n] = 1$. If $\norm{X_i-X_j} \in (r_n,2r_n]$, then both $i \in \M_{r_n}$ and $j \in \M_{r_n}$ if and only if the other locations in the $i$ and $j$'s $r_n$-balls all have larger marks than $U_i,U_j$. That is,
  \begin{multline*}
    E \equiv \big\{i\in\M_{r_n} \medcap j\in\M_{r_n}\big\} = \big\{U_i<U_k \ \forall k \text{ s.t. } X_k\in B(i,r_n)\setminus\{i\} \\ \medcap U_j<U_k \ \forall k \text{ s.t. } X_k\in B(j,r_n)\setminus\{j\}\big\}.
  \end{multline*}

  \noindent Define
  \begin{equation*}
    m_{ij} = \abs{B(X_i,r_n) \cap B(X_j,r_n) \cap \X_n}.
  \end{equation*}

  \noindent Because the marks are i.i.d.\ and continuously distributed, all $(N_i+N_j-m_{ij})!$ orderings of the marks in $B(i,r_n)\cup B(j,r_n)$ are equally likely. If $U_i<U_j$, then every point of $B(i,r_n)$ has mark larger than $U_i$, while every point of $B(j,r_n)$ has mark larger than $U_j>U_i$, so $U_i$ is the smallest mark. Hence 
  \begin{equation*}
    \prob\big(E \medcap U_i<U_j \mid \X_n\big) = \frac{1}{\binom{N_i+N_j-m_{ij}}{1}}\frac{1}{\binom{N_j}{1}},
  \end{equation*}
  
  \noindent and similarly for $U_j<U_i$, so
  \begin{equation*}
    \prob(E\mid \X_n) = \frac{1}{N_j (N_i+N_j-m_{ij})} + \frac{1}{N_i(N_i+N_j-m_{ij})} = \frac{N_i+N_j}{N_iN_j(N_i+N_j-m_{ij})}.
  \end{equation*}
  
  \noindent Therefore,
  \begin{equation}
    N_i N_j\,\E[M_i M_j\mid\X_n] = \frac{N_i+N_j}{N_i+N_j-m_{ij}} \leq 2. \label{Se2}
  \end{equation}

  \noindent {\bf Case $\abs{S}=3$.} Take $S = \{i,j,k\}$. If the distance between some pair in $S$ is at most $r_n$, then $M_iM_jM_k=0$ a.s., so assume all pairwise distances exceed $r_n$. Let 
  \begin{equation*}
    m_{ijk} = \abs{B(X_i,r_n)\cap B(X_j,r_n)\cap B(X_k,r_n) \cap \X_n}.
  \end{equation*}

  \noindent For each $\{\ell,m\}\subseteq\{i,j,k\}$, define $\tilde{N}_\ell\equiv N_\ell-m_{ijk}$ and $\tilde{m}_{\ell m}\equiv m_{\ell m}-m_{ijk}$. Let $a_i'=\tilde{N}_i-1-\tilde{m}_{ij}-\tilde{m}_{ik}$, the number of locations exclusive to $B(i,r_n)$. Then $\int_0^1 (1{-}u)^{a_i'} \,\text{d}u$ is the chance that $i$ has the lowest mark of all such locations, and
  \begin{equation*}
    \prob(M_iM_jM_k=1\mid\X_n) \leq \prod_{\ell \in \{i,j,k\}} \int_0^1 (1 - u)^{a_\ell'} \,\text{d}u = \prod_{\ell \in \{i,j,k\}} (a_\ell'+1)^{-1} \quad\text{a.s.} 
  \end{equation*}

  \noindent Therefore,
  \begin{align*}
    N_iN_jN_k\,\prob(M_iM_jM_k=1\mid\X_n) &\leq \prod_{\substack{s \subseteq \{i,j,k\} \\ \{p,q\} \subseteq \{i,j,k\}\backslash\{s\}}} \frac{N_s}{\tilde{N}_s-\tilde{m}_{sp}-\tilde{m}_{sq}} \\
					  &\leq \prod_{\substack{s \subseteq \{i,j,k\} \\ \{p,q\} \subseteq \{i,j,k\}\backslash\{s\}}} \frac{N_s}{N_s-m_{sp}-m_{sq}},
  \end{align*}

  \noindent which is uniformly $O(1)$ a.s.\ by \autoref{lNm}.

  \bigskip
  \noindent {\bf Case $\abs{S}=4$.} Take $S = \{i,j,k,\ell\}$. If the distance between some pair in $S$ is at most $r_n$, then $M_iM_jM_kM_\ell=0$ a.s., so suppose all pairwise distances exceed $r_n$. Define $\tilde{N}_i\equiv N_i-m_{ijk}-m_{ij\ell}-m_{ik\ell}+2m_{ijk\ell}$, $\tilde{m}_{ij}\equiv m_{ij}-m_{ijk}-m_{ij\ell}+m_{ijk\ell}$, and $\tilde{a}_\ell=\tilde{N}_\ell-\tilde{m}_{i\ell}-\tilde{m}_{j\ell}-\tilde{m}_{k\ell}-1$. Then
  \begin{equation*}
    \prob(M_iM_jM_kM_\ell=1\mid\X_n) \leq \prod_{s \in \{i,j,k,\ell\}} \int_0^1 (1 - u)^{\tilde{a}_s} \,\text{d}u = \prod_{s \in \{i,j,k,\ell\}} \tilde{a}_s^{-1} \quad\text{a.s.}
  \end{equation*}

  \noindent Therefore
  \begin{multline*}
    N_iN_jN_kN_\ell\,\prob(M_iM_jM_kM_\ell=1\mid\X_n) \leq \prod_{\substack{s \in \{i,j,k,\ell\} \\ \{p,q,r\} \subseteq \{i,j,k,\ell\}\backslash\{s\}}} \frac{N_s}{\tilde{N}_s-\tilde{m}_{sp}-\tilde{m}_{sq}-\tilde{m}_{sr}} \\
    \leq \prod_{\substack{s \in \{i,j,k,\ell\} \\ \{p,q,r\} \subseteq \{i,j,k,\ell\}\backslash\{s\}}} \frac{N_s}{N_s-m_{sp}-m_{sq}-m_{sr}},
  \end{multline*}

  \noindent which is uniformly $O(1)$ a.s.\ by \autoref{lNm}.
\end{proof}

\begin{seclemma}\label{lovar}
  Under \autoref{aboundY}, both $\var(\sum_{i=1}^n Z_i^* \mid \X_n)$ and $\sum_{i=1}^n \var(Z_i^* \mid \X_n)$ are $O(n^2\beta_n^d)$ a.s.
\end{seclemma}
\begin{proof}
  Observe that
  \begin{equation*}
    \var(Z_i^* \mid \X_n) = N(i,r_n) \left(\frac{\mu(X_i,\{X_i\})^2}{p} + \frac{\mu(X_i,\emptyset)^2}{1-p}\right) - \tau(X_i)^2,
  \end{equation*}

  \noindent which is $O(N(i,r_n))$ by \autoref{aboundY}, which is in turn uniformly $O(n\beta_n^d)$ a.s.\ by \eqref{Norder}.

  The covariance term is
  \begin{align*}
    \cov(Z_i^*, Z_j^* \mid \X_n) = \left(\frac{\prob(i \in \M_{r_n},\, j \in \M_{r_n} \mid \X_n)}{\rho_i\rho_j} - 1\right)\tau(X_i) \tau(X_j).
  \end{align*}

  \noindent If $\norm{X_i - X_j} \leq r_n$, then $\prob(i \in \M_{r_n}, j \in \M_{r_n} \mid \X_n) = 0$ and $\cov(Z_i^*, Z_j^* \mid \X_n) = -\tau(X_i) \tau(X_j)$. If $\norm{X_i - X_j} > 2r_n$, then $i \in \M_{r_n} \indep j \in \M_{r_n} \mid \X_n$ and $\cov(Z_i^*, Z_j^* \mid \X_n) = 0$. If $r_n < \norm{X_i - X_j} \leq 2r_n$, then by \eqref{Se2},
  \begin{equation*}
    \cov(Z_i^*, Z_j^* \mid \X_n) = \frac{m_{ij}}{N(i,r_n)+N(j,r_n)-m_{ij}}\,\tau(X_i) \tau(X_j)
  \end{equation*}

  \noindent for $m_{ij} \equiv \abs{B(X_i,r_n) \cap B(X_j,r_n) \cap \X_n}$.

  Using the identities
  \begin{align*}
    &\frac{\mu(X_i,\{X_i\})^2}{p} + \frac{\mu(X_i,\emptyset)^2}{1-p} = \frac{((1-p)\mu(X_i,\{X_i\}) + p\mu(X_i,\emptyset))^2}{p(1-p)} + \tau(X_i)^2 \quad\text{and} \\
    &\sum_{i=1}^n (N(i,r_n)-1) \tau(X_i)^2 - \sum_{\substack{j\neq i \\ \norm{X_i-X_j}\leq r_n}} \tau(X_i) \tau(X_j) = \sum_{\substack{i < j \\ \norm{X_i-X_j}\leq r_n}} (\tau(X_i) - \tau(X_j))^2,
  \end{align*}

  \noindent the previous derivations yield
  \begin{multline}
    \var\left(\sum_{i=1}^n Z_i^* \,\bigg|\, \X_n\right) = \frac{1}{p(1-p)} \sum_{i=1}^n N(i,r_n) \big((1-p)\mu(X_i,\{X_i\}) + p\mu(X_i,\emptyset)\big)^2 \\ + \sum_{\substack{i < j \\ \norm{X_i-X_j}\leq r_n}} (\tau(X_i) - \tau(X_j))^2 + \sum_{\substack{j \neq i \\ r_n < \norm{X_i-X_j} \leq 2r_n}} \frac{m_{ij}}{N(i,r_n)+N(j,r_n)-m_{ij}} \,\tau(X_i) \tau(X_j). \label{s2e}
  \end{multline}

  \noindent The first two terms are non-negative, while the third can have either sign. Also $m_{ij}/(N(i,r_n)+N(j,r_n)-m_{ij}) \leq 1$ since $m_{ij} \leq \min\{N(i,r_n), N(j,r_n)\}$. Then by \autoref{aboundY} and \eqref{Norder}, the right-hand side is $O(n^2\beta_n^d)$ a.s. 
\end{proof}

\begin{seclemma}\label{lstein}
  Let $\{Z_i\}_{i=1}^n$ be random variables with dependency graph $\bm{A}$ such that $\E[Z_i^4]<\infty$ and $\E[Z_i]=0$.\footnote{For a graph $\bm{A}$, let $A_{ij}$ be an indicator for whether nodes $i,j$ are connected in $\bm{A}$. We say $\bm{A}$ is a {\em dependency graph} for $\{Z_i\}_{i=1}^n$ if for any $S_1,S_2 \subseteq \{1,\ldots,n\}$ such that $A_{k\ell}=0$ for all $k \in S_1, \ell \in S_2$, we have $\{Z_k\}_{k \in S_1} \indep \{Z_\ell\}_{\ell \in S_2}$.} Define $\sigma^2 \equiv \var(\sum_{i=1}^n Z_i)$ and $\mathcal{W} = \sum_{i=1}^n Z_i/\sigma$. For $\mathcal{Z} \sim \mathcal{N}(0,1)$,
  \begin{equation}
    d(\mathcal{W},\mathcal{Z}) \leq \frac{1}{\sigma^3} \sum_{i=1}^n \sum_{j=1}^n \sum_{k=1}^n \E[\abs{Z_iZ_jZ_k}] A_{ij} A_{ik} + \frac{\sqrt{2}}{\sqrt{\pi}\sigma^2} \var\left( \sum_{i=1}^n \sum_{j=1}^n Z_iZ_j A_{ij} \right)^{1/2}, \label{wass2}
  \end{equation}

  \noindent where $d(\cdot,\cdot)$ is the Wasserstein distance.
\end{seclemma}
\begin{proof}
  This follows from \cite{ross2011fundamentals} Theorem 3.1 and the bounds in his equations (3.9)--(3.12).
\end{proof}

%--------------------------------------
\subsection{Boundary}
%--------------------------------------

\begin{secdefinition}[\cite{adams2003sobolev}, 4.9]\label{dlipbd}
  A bounded set $\Omega \subset \R^d$ for $d\geq 2$ has a {\em Lipschitz boundary} if there exist $\delta,M>0$, a locally finite open cover $\{U_j\}$ of $\Omega$, and, for each $j$, a function $f_j\colon \R^{d-1} \rightarrow \R$ such that the following hold.
  \begin{enumerate}[(a)]
    \item For some finite $R$, every collection of $R+1$ of the sets $U_j$ has empty intersection.
    \item Let $\partial \Omega$ denote the boundary of $\Omega$ and $\Omega_\delta = \{x \in \Omega\colon \norm{x-y} < \delta, \forall y \in \partial\Omega\}$. For every $x,y \in \Omega_\delta$ such that $\norm{x-y} < \delta$, there exists $j$ such that $x,y \in \{z \in U_j\colon \norm{z-w} < \delta, \forall z \in \partial U_j\}$.
    \item There exists $M>0$ such that each $f_j$ is Lipschitz with constant $M$.
    \item For some Cartesian coordinate system $(\zeta_{j,1}, \ldots, \zeta_{j,n})$ in $U_j$, $\Omega \cap U_j$ is represented by the inequality $\zeta_{j,n} < f_j(\zeta_{j,1}, \ldots, \zeta_{j,n-1})$.
  \end{enumerate}
\end{secdefinition}

\noindent This definition states that each point $x$ on the boundary of $\Omega$ has a neighborhood $U_x$ whose intersection with the boundary is the graph of a Lipschitz function.

\begin{secdefinition}[\cite{adams2003sobolev}, 4.6]\label{dcc}
  A bounded set $\Omega \subset \R^d$ for $d\geq 2$ satisfies the {\em cone condition} if there exists a finite cone (with height $\rho$, axis direction $v$, aperture angle $\kappa$, and vertex at the origin) 
  \begin{equation*}
    C \equiv C(\rho, v, \kappa) = \left\{ x \in \R^d \colon x=0 \text{ or } 0<\abs{x}<\rho, \angle(x,v) \leq \kappa/2 \right\},
  \end{equation*}

  \noindent where $\angle(x,v)$ is the angle between position vector $x$ and $v$, such that each $x \in \Omega$ is the vertex of a finite cone $C_x$ contained in $\Omega$ and congruent to $C$.
\end{secdefinition}

\begin{seclemma}\label{lballprob}
  Suppose $\text{supp}(f)$ has a Lipschitz boundary, and $\underline{f} = \inf_{x\in\text{supp}(f)} f(x) > 0$. Then there exists $c_0 \in (0,1]$ and $\rho > 0$ such that for any $x,x' \in \text{supp}(f)$ and $r,r' \geq 0$,
  \begin{equation*}
    \prob(\tilde{X}_1 \in B(x,r) \cap B(x',r')) \geq \underline{f}\, c_0 \text{Vol}(B(x,\min\{r,\rho\}) \cap B(x',\min\{r',\rho\})).
  \end{equation*}
\end{seclemma}
\begin{proof}
  First,
  \begin{equation*}
    \prob(\tilde{X}_1 \in B(x,r) \cap B(x',r')) \geq \underline{f}\, \text{Vol}(B(x,r) \cap B(x',r') \cap \text{supp}(f)).
  \end{equation*}

  \noindent By 4.11 of \cite{adams2003sobolev}, if a set has a Lipschitz boundary, then it satisfies the cone condition for some finite cone $C(\rho, v, \kappa)$. Then for some rotation $\bm{R}_x\colon \R^d\rightarrow\R^d$, $B(x,r) \cap \text{supp}(f)$ contains the cone $x + \bm{R}_x(C(\min\{\rho,r\} v, \kappa))$ whose volume is a fixed fraction $c_0$ of $\text{Vol}(B(x,\min\{r,\rho\}))$ for $c_0 \in (0,1]$ independent of $x,r$. Hence $\text{Vol}(B(x,r) \cap B(x',r') \cap \text{supp}(f)) \geq c_0 \text{Vol}(B(x,\min\{r,\rho\}) \cap B(x',\min\{r',\rho\}))$.
\end{proof}

%----------------------------------------------------------------------

\FloatBarrier
\phantomsection
\addcontentsline{toc}{section}{References}
\bibliography{space_treats}{} 
\bibliographystyle{aer}

%----------------------------------------------------------------------

\end{document}